\documentclass[twoside]{amsart}
\usepackage{latexsym}
\usepackage{amssymb,amsmath,amsopn}
\usepackage[dvips]{graphicx}   
\usepackage{color,epsfig}      
\usepackage{xcolor}
\usepackage{cancel}

\numberwithin{equation}{section}

               {\begin{list}{}{\leftmargin#1\rightmargin#2}\item{}}%
               {\end{list}}

\newenvironment{pfof}[1]{\vspace{1ex}\noindent{\em Proof of
#1}\hspace{0.5em}}{\hfill\qed\vspace{1ex}}

\newtheorem{thm}{Theorem}
\newtheorem{ass}{Assumption}
\newtheorem{lem}[thm]{Lemma}
\newtheorem{prop}[thm]{Proposition}
\theoremstyle{definition}
\newtheorem{defn}{Definition}
\newtheorem{rem}[thm]{Remark}

\newtheorem{claim}[thm]{Claim}

\renewcommand{\Re}{\mathbb R}

\def\bea{\begin{eqnarray}}
\def\eea{\end{eqnarray}}
\def\d{\Delta}
\def\bF{\mathbf{F}}
\def\bN{\mathbf{N}}
\def\bn{\mathbf{n}}
\def\bV{\mathbf{V}}
\def\bT{\mathbf{T}}
\def\bX{\mathbf{X}}
\def\bx{\mathbf{x}}
\def\bW{\mathbf{W}}
\def\btau{\boldsymbol{\tau}}
\def\bxi{\boldsymbol{\xi}}
\def\bta{\boldsymbol{\eta}}
\def\ux{\underline{x}}
\def\uX{\underline{X}}
\def\ubx{\underline{\mathbf{x}}}

\def\Plim{\stackrel{\mathbb{P}}{\Longrightarrow}}

\begin{document}


\title[An evolution model for crack patterns]{An evolution model for polygonal tessellations as models for crack networks and other natural patterns}
\author[P. B\'alint, G. Domokos and K. Reg\H os ] {P\'eter B\'alint, G\'abor Domokos and Krisztina Reg\H os}
\address{P\'eter B\'alint, ELKH-BME Stochastics Research Group, and  Department of Stochastics, Institute of Mathematics, Budapest University of Technology and Economics,
M\H{u}egyetem rkp. 3, H-1111, Budapest, Hungary}
\email{pet@math.bme.hu}
\address{G\'abor Domokos, MTA-BME Morphodynamics Research Group and Dept. of Mechanics, Materials and Structures, Budapest University of Technology and Economics,
M\H uegyetem rakpart 1-3., Budapest, Hungary, 1111}
\email{domokos@iit.bme.hu}
\address{Krisztina Reg\H os,MTA-BME Morphodynamics Research Group, Budapest University of Technology and Economics,
M\H uegyetem rakpart 1-3., Budapest, Hungary, 1111}
\email{regoskriszti@gmail.com}
\thanks{Support of the NKFIH Hungarian Research Fund, grants 134199, 142169 and 144059, and of the NKFIH Fund TKP2021 BME-NVA, carried out at the Budapest University of Technology and Economics, is kindly acknowledged. Krisztina Reg\H os: This research has been supported by the program UNKP-22-3 by ITM and NKFIH. The gift representing the Albrecht Science Fellowship is gratefully appreciated.
}


\begin{abstract}
We introduce and study a general framework for modeling the evolution of crack networks. The evolution steps are triggered by exponential clocks corresponding to local micro-events, and thus reflect the state of the pattern. In an appropriate simultaneous limit of pattern domain tending to infinity and time step tending to zero, a continuous time model, specifically a system of ODE is derived that describes the dynamics of averaged quantities. In comparison with the previous, discrete time model, studied recently by two of the present three authors, this approach has several advantages. In particular, the emergence of non-physical solutions characteristic to the discrete time model is ruled out in the relevant nonlinear version of the new model. We also comment on the possibilities of studying further types of pattern formation phenomena based on the introduced general framework.

\end{abstract}
\maketitle
\tableofcontents
\section{Introduction and motivation}\label{sec:intro}
Crack networks are not only beautiful natural patterns, their geometry tells the story of the geophysical processes which have
created them. No wonder that there exists substantial literature on this subject \cite{adler_fracture_network_book,columnar_joint_science_1988, goering_columnar_2008,  haltigin2012geometric} which aims to describe the geometry of these patterns as polygonal tessellations of the Euclidean plane. While the theory of the latter is well established in its own right \cite{grunbaumshepard, senechal}, the connection to crack networks was established by a mean field theory \cite{Gray_1976_irregular_mosaics_formula} which was subsequently formalized and extended \cite{domokos2019honeycomb, schneider2008stochastic}.  This theory helped to create a meaningful classification of natural crack patterns \cite{Plato} which could be even extended into the third dimension. Polygonal tessellations offer not only a good model for many crack networks, the same model can be also used to describe a much broader
variety of ground patterns.  Figure \ref{fig:00} illustrates these patterns, the first two being crack mosaics, the third a
ground pattern of different origin but equally well approximated by a polygonal tessellation.

\begin{figure}[h!]
\begin{center}
\includegraphics[width=\textwidth]{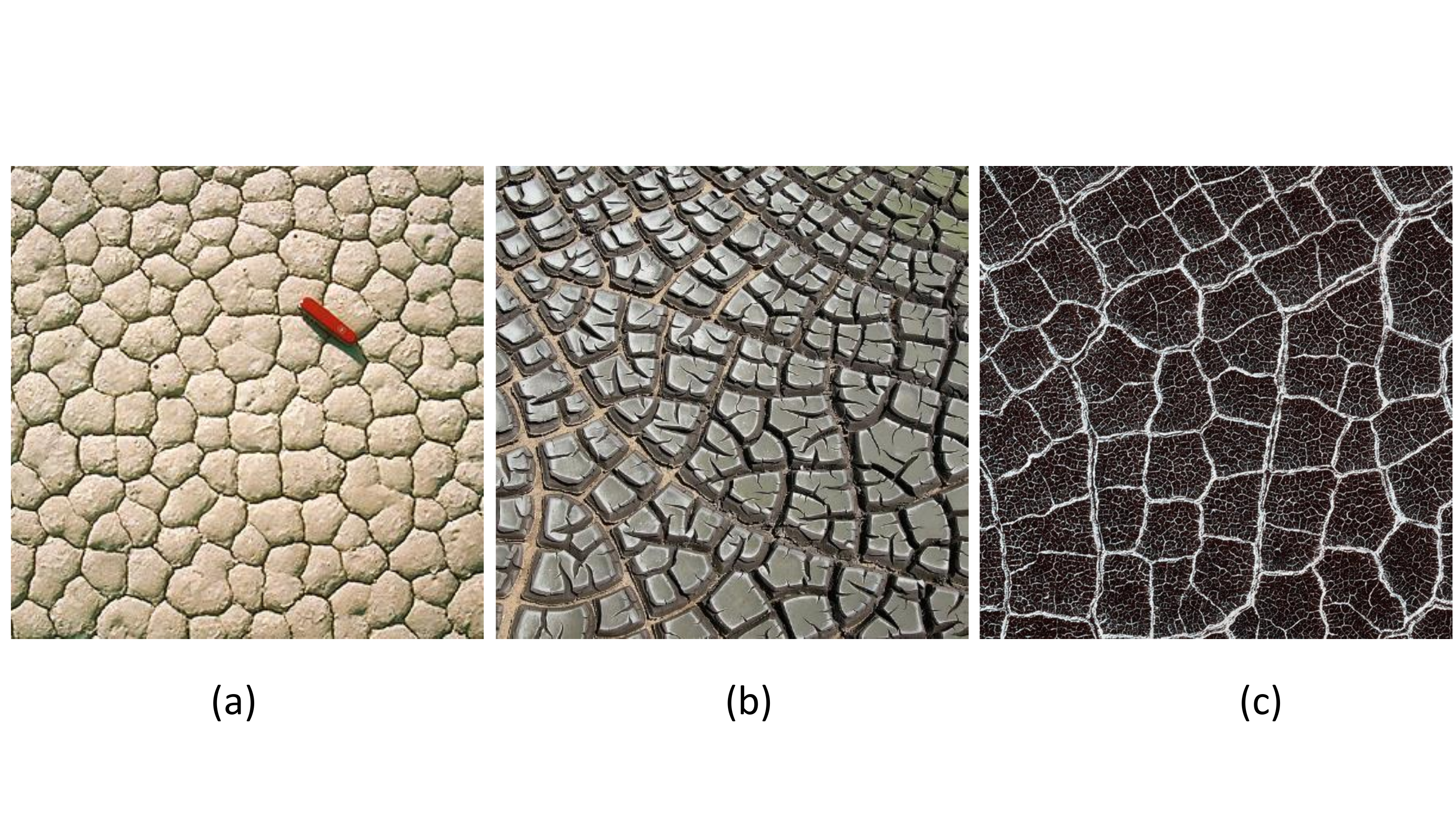}
\caption{Polygonal ground patterns. (a) and (b) Desiccation crack patterns in mud. (c) Polygonal patterned ground on Mars.  Photo credit (a) Charles E. Jones University of Pittsburgh, Pittsburgh, PA
(b) Hannes Grobe, Alfred Wegener Institute, Bremerhaven, Germany (c) NASA/JPL}\label{fig:00}
\end{center}
\end{figure}

Most of the mentioned results relate to the \emph{static} characterization of  pattern geometry. However, in \cite{Plato} the idea was put forward that perhaps the same tools could be used to build a dynamical theory, describing the time evolution of these patterns. The experimental results in  \cite{goering_evolv_pattern_2013} paved the way for the first such model which was introduced in \cite{disc} with the goal to find potential \emph{evolution paths} connecting crack networks of related geological origins. Figure \ref{fig:00}(a) and (b) illustrate two such patterns, observed in drying mud. These two patterns were used in \cite{disc}  to illustrate all basic concepts and the application of the evolution model. In this article we will use the very same examples, thus offering a direct comparison to the results in \cite{disc}, where a discrete time evolution model was introduced, driven by local micro-events associated with constant, a priory probabilities. In this paper we will refer to the evolution model introduced in \cite{disc} as the \textit{discrete model}. It is revisited in the form of Assumptions~\ref{ass:1} and~\ref{ass:2} in Section~\ref{sec:discmodel}. While the discrete model shows good match with data, it has several disadvantages which can be summarized as follows:
\begin{itemize}
\item We aim to describe the continuous time evolution of density type quantities such as the average degree of nodes or faces, which  actually make sense in a thermodynamic limit of infinite (or growing) patterns. As the consecutive steps of the evolution correspond to local micro-events, in particular to the cracking of faces or the healing of irregular nodes, it is expected that as growing patterns are considered, the process accelerates. Accordingly, we seek for a description that simultaneously incorporates the thermodynamic limit of infinite patterns and the transition to continuous time. The discrete model, however, has no notion of time built in.
\item In the discrete model, the probabilities of the different types of steps are represented by rigid quantities, which thus cannot reflect the current state of the mosaic. This is an unrealistic feature as it can be expected that, for example, as the number of irregular nodes decreases, there are less options for healing rearrangements, which thus decreases the likelihood of such steps. While it is true that state-dependent probabilities could be also introduced in the discrete model, this was not done in \cite{disc}.
\item Closely related to the above mentioned issues, as already observed in \cite{disc}, the discrete model  admits the emergence of some nonphysical solutions.
\end{itemize}
We will discuss these features of the discrete model in Sections \ref{sec:discmodel} and \ref{sec:return}.
The mathematical goal of the current paper is threefold:
\begin{enumerate}
    \item We introduce a new evolution model by replacing Assumption~\ref{ass:2} of the discrete model by the new Assumption~\ref{ass:3}. The main new idea is to tie the micro-events in Assumption \ref{ass:3} to \emph{exponential clocks} with specific geometric locations. On top of laying a more rigorous mathematical foundation, this way we aim to provide a possibility for the systematic elimination of the nonphysical solutions the discrete model produced.
    \item The set of 3 variables used in \cite{disc} is extended to a system of arbitrary number of variables. This is not only a quantitative change: the new framework would include models of physical processes which are conceptually fundamentally different from the one discussed in \cite{disc}.
    \item In Theorem~\ref{lem:1} we derive the rigorous limit, resulting in the general \textit{governing ordinary differential equation} (\ref{ode1}) for the evolution of infinite network patterns, driven by local events. In order to connect the general ODE with specific physical processes, we define the \emph{fundamental table} associated with the latter (Definition \ref{def:fund}).
\end{enumerate}
Although the differential equation (\ref{ode1}) in Theorem \ref{lem:1} includes all three elements of the above list, let us remark that
 they are, nevertheless, independent: any of these three constructions could be realized without the other two. Accordingly, we believe that our new model is a good candidate for the realistic description of various types of crack evolution processes as it addresses the above three points simultaneously.
Beyond these goals, we aim to demonstrate the power of this general framework by computing the fundamental tables for the model in \cite{disc} as well as for its extended version with exponential clocks. Using the fundamental tables we also derive their respective governing ODEs to which we will refer as the \emph{linear model} and the \emph{nonlinear model}.

The structure of the paper is the following: in Section \ref{sec:discmodel} we describe the assumptions and the discrete time model presented in \cite{disc}. The mathematical core of the paper is Section \ref{sec:math} where we aim to achieve the goals (1)-(3), listed above. In Section \ref{sec:return} we return to the discrete model of \cite{disc}, fit it into the expanded framework presented in the previous section and derive both the corresponding linear system (which is the continuous time version of the model in \cite{disc}) and also
\textit{the nonlinear system which is free of nonphysical solutions}. We derive the fundamental tables for both cases and, based on the latter, we compute the governing ODE. In Section \ref{sec:sum}
we compare the trajectories computed for the linear and nonlinear models. We also draw some geological conclusions and describe potential application areas.

\section{Description of the discrete time model presented in \cite{disc}}\label{sec:discmodel}
\subsection{Static assumptions and definitions}
The discrete model in \cite{disc} is based on the following assumption:
\begin{ass}\label{ass:1}
 Let $M$ be a convex, normal, balanced mosaic \cite{grunbaumshepard} and let $M_d$ denote a finite domain of $M$, defined by a ball with diameter $d$. Let $M_d$ have $F(=F_d)$ faces (cells) $V(=V_d)$ nodes (vertices). We assume that $M_d$  is a geometric model of a crack network.
\end{ass}

We remark that all cells of a convex mosaic are convex polygons \cite{schneider2008stochastic}.

\begin{defn}\label{def:1} \cite{regos2022twovertex}
Let $M_d$ be a finite domain of a balanced, convex mosaic (as noted in Assumption \ref{ass:1}) and let $M_d$
have $F$ faces (cells) and $V$ nodes (vertices), i.e. points where more than 2 cells overlap.
We say that a (polygonal) cell $f_i$, $(i=1,2, \dots F)$ has a \emph{regular vertex} on its boundary at point $P_f$,  if two straight edges of $f_i$ meet
at $P_f$ at an angle different from $\pi$. In addition, $f_i$ may have \emph{irregular vertices} in the interior
of its edges at points where more than 2 cells overlap. For a cell (face) $f_i$,$(i=1,2, \dots F)$ let $v^{\star}_i$ denote the number of its regular vertices, $v_i\geq v^{\star}_i$ the total number of its vertices and we call $v_i$ and $v^{\star}_i$ respectively the \emph{combinatorial degree} and \emph{corner degree} of the $i$th face. We use the same counting rule at nodes and we call $n_i$ and $n^{\star}_i$ $(i=1,2, \dots V)$ respectively the \emph{combinatorial degree} and \emph{corner degree} of the $i$th node. We call the $i$th node regular if $n_i=n^{\star}_i$, we denote the number of irregular nodes  by $V_I$ and we call
$$ r(M_d)=\frac{V-V_I}{V}$$ the regularity of $M_d$.
Using these concepts we also introduce
\[
N_F=\sum_{i=1}^F v_i ;\qquad N_V=\sum_{j=1}^V n_j
\]
and
\[
N^{\star}_F=\sum_{i=1}^F v^{\star}_i ;\qquad N^{\star}_V=\sum_{j=1}^V n^{\star}_j.
\]
\end{defn}
We remark that the difference between  $N_F$ and $N_V$ (respectively,  between $N^{\star}_F$ and $N^{\star}_V$) is only due to the boundary; on domains with no boundary this difference vanishes.  We are interested in the $d\to \infty$ limit and since, according to Assumption \ref{ass:1}, $M$ is a
balanced mosaic \cite{grunbaumshepard}, the respective differences between $N_F$ and $N_V$
and between $N^{\star}_F$ and $N^{\star}_V$ will vanish in this limit, in particular, we have
\begin{equation}
    \lim _{d \to \infty}\frac{N_F}{N_V}=\lim _{d\to \infty}\frac{N^{\star}_F}{N^{\star}_V}=1.
\end{equation}
To simplify notation, we introduce
\begin{equation}
    N=(N_F+N_V)/2, \qquad N^{\star}=(N^{\star}_F+N^{\star}_V)/2.
\end{equation}
In the context of geophysical crack networks, corner degrees appear to be more relevant than combinatorial degrees, so
henceforth we will only track the evolution of $F,V$ and $N^{\star}$ which we call the \emph{fundamental variables} and we observe that under Assumption \ref{ass:1} we have
\begin{equation}\label{npropto}
    (F,V,N^{\star}) \asymp d^2
\end{equation}
since $M$ is assumed to be a balanced mosaic \cite{grunbaumshepard}.
Here, and throughout the paper, given some quantity $R\to\infty$, $f(R) \asymp g(R)$ means that, there exists some $C>0$ independent of $R$ such that $C^{-1} g(R) \le f(R)\le C g(R)$ for any $R$.

\subsection{Representation of mosaics}
To record and represent mosaics, it is plausible to use dimensionless variables.
In \cite{disc}
\begin{equation}\label{variables}
    \bar n ^{\star}=\frac{N^{\star}}{V}, \quad \bar v^{\star}=\frac{N^{\star}}{F}
\end{equation}
have been introduced as the \emph{average nodal corner degree } and \emph{average cell corner degree} as such variables
and the $[\bar n^{\star},\bar v^{\star}]$ plane is referred to as the \emph{symbolic plane} of mosaics \cite{domokos2019honeycomb, Plato}.
Alternatively, we may use the reciprocal variables
\begin{equation}\label{inversesymbolic}
    x =\frac{V}{N^{\star}}, \quad y=\frac{F}{N^{\star}}
\end{equation}
and we will refer to $[x,y]$ as the \emph{inverse symbolic plane}.

\begin{rem}\label{convdomain}
Both representations are illustrated in Figure \ref{fig:0}. One advantage of the application of the inverse symbolic plane is that the \emph{domain of convex mosaics} \cite{domokos2019honeycomb} is defined by the four straight lines $L_i$:
\begin{equation}\label{eq:Q}
\begin{array}{rrcl}
L_1: & y & = & x/2, \\
L_2: & y & = & 1/2 - x, \\
L_3: & y & = & 1/3, \\
L_4: & y & = & 1/2 - x/2.
\end{array}
\end{equation}
Throughout the paper, this domain will be denoted by the symbol $Q$.
\end{rem}

\begin{rem}\label{linearmodel}
In Section \ref{sec:math} we will show that the trajectories
of the ODE corresponding to the model based on Assumptions \ref{ass:1} and \ref{ass:2} appear as straight lines in the inverse symbolic
plane. This fact was the motivation to call this the \emph{linear model.}
\end{rem}

\begin{figure}[h!]
\begin{center}
\includegraphics[width=1.2\textwidth]{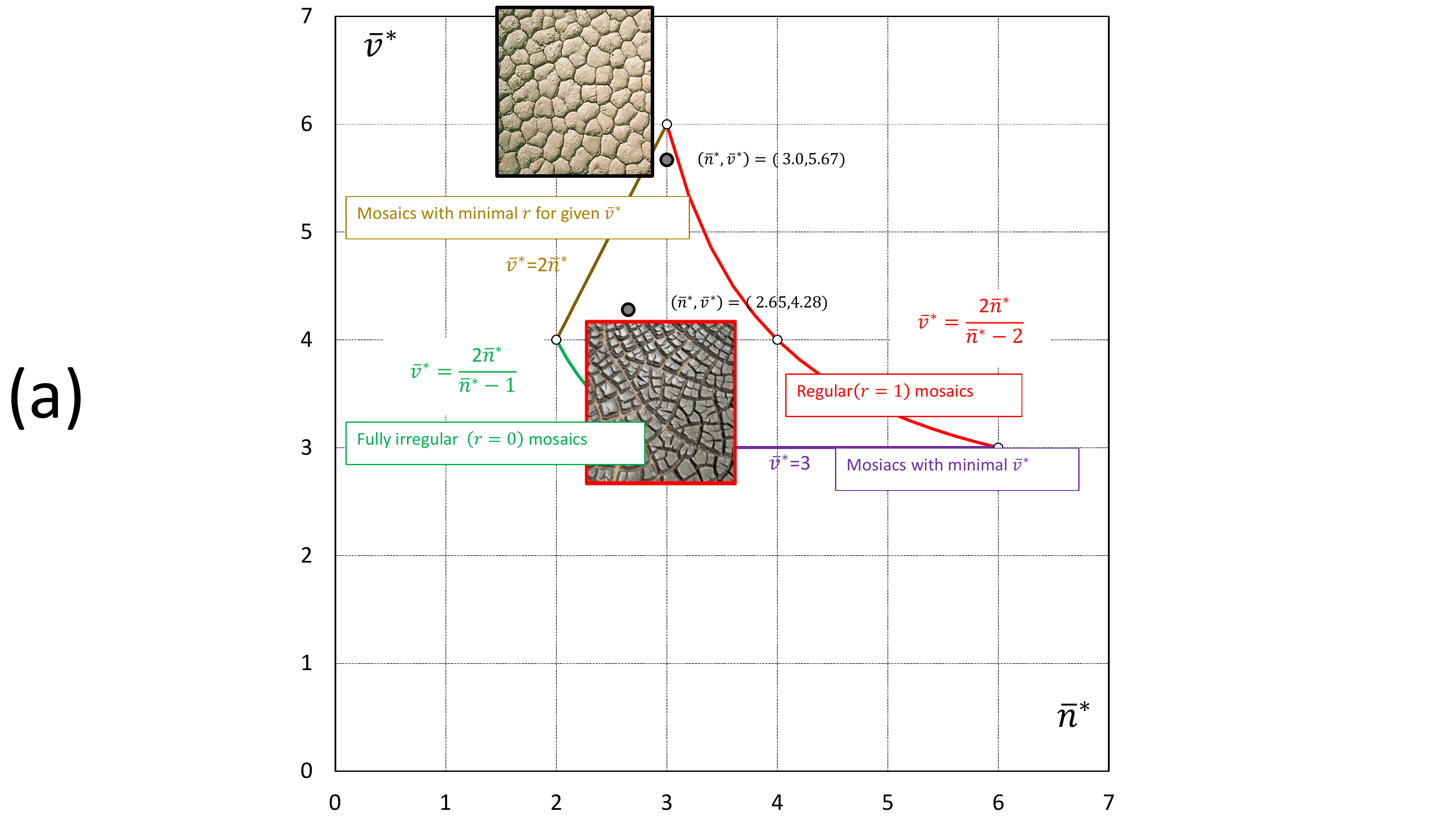}
\includegraphics[width=1.2\textwidth]{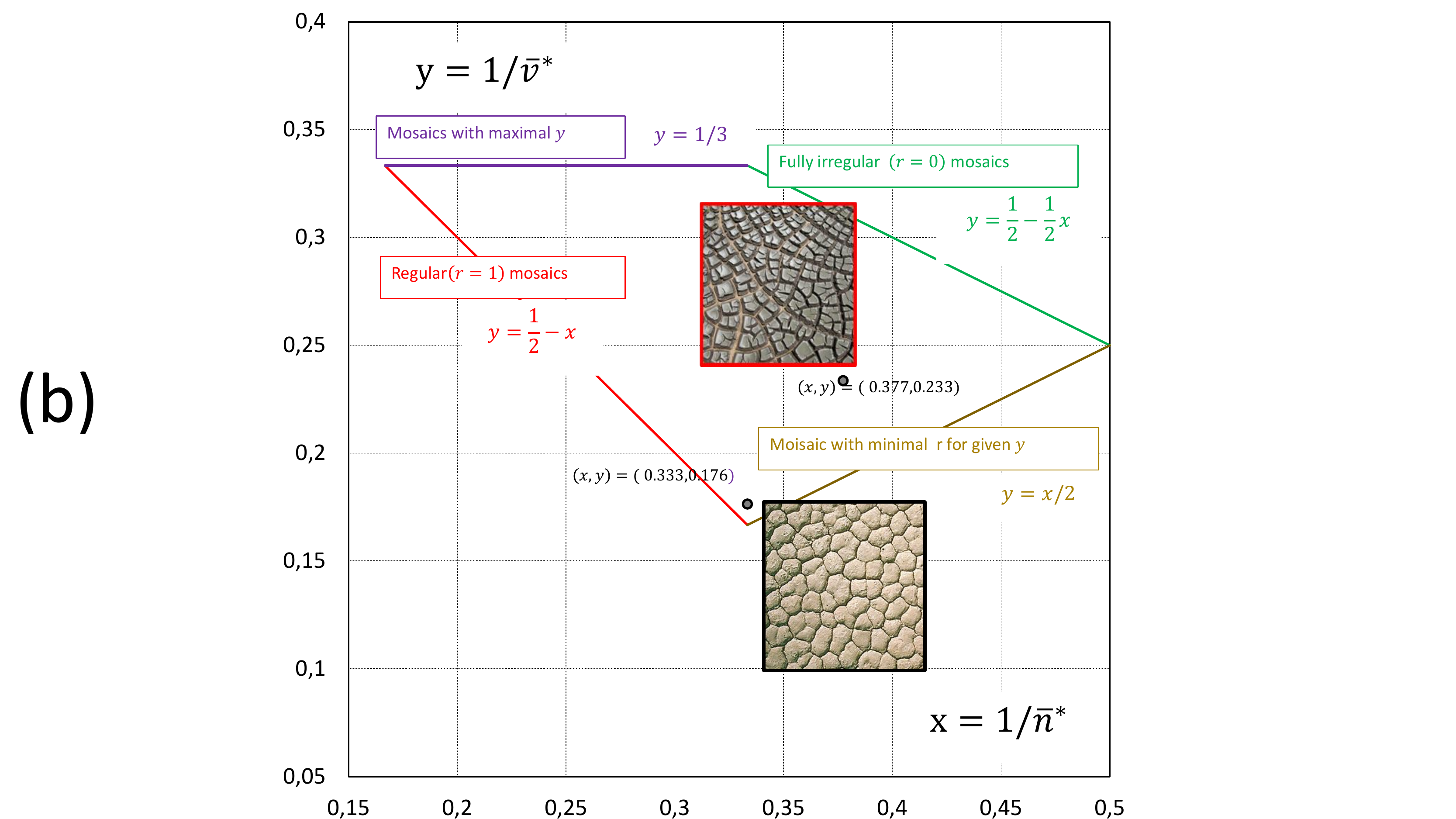}
\caption{Illustration (a)  of the symbolic plane and (b) of the inverse symbolic plane. The region bounded by the red, purple, green and brown boundary curves (lines in case (b)) corresponds to physical solutions.}\label{fig:0}
\end{center}
\end{figure}

\subsection{Dynamic assumptions and definitions}
We may call Assumption \ref{ass:1} a \emph{static} assumption, defining the model of the crack network for any fixed time $t=t_0$. Now we include the model of evolution and we formulate, based on \cite{disc} a set of hypotheses to which we will jointly refer as the \emph{discrete model}. In Section \ref{sec:math}, we will embed this model in a more general framework as a \textit{linear model}. The later name refers to the fact that  the trajectories of this linear model appear as straight lines in the special representation introduced in equation \eqref{inversesymbolic}, called the \emph{inverse symbolic plane}, see also Remark \ref{linearmodel}. In Section \ref{sec:return} we will return to the current model and explain its linearity in full depth.

\begin{ass}\label{ass:2} (Discrete model)
We will assume that the evolution of the finite mosaic $M_d$ defined in Assumption \ref{ass:1} and Definition \ref{def:1} is driven by two discrete micro-events which we call $R_0,R_1$, respectively. We describe these steps below.
\begin{enumerate}
\item During \emph{secondary cracking}, one cell of the primary crack network is split into two parts along a straight line segment, connecting two points belonging to the relative interior of two different edges of the cell. ($R_0$-type step). It is easy to see that $R_0$-type steps retain the convexity of the initial mosaic. See Figure \ref{fig:2}.
\item During \emph{crack healing-rearrangement}, the edges and nodes of the crack network are rearranged so that 'T' nodes evolve into 'Y' nodes \cite{goering_evolv_pattern_2013} and each such event corresponds to an $R_1$-type step.  See Figure \ref{fig:2}. The convexity of the mosaic is retained also in this step.
\item The micro-events $R_0$ and $R_1$ occur randomly with respective probabilities $p_0=(1-q)$ and $p_1=q$. The parameter $q$ is a global constant in the evolution process.\footnote{This rule is implicitly assuming that in a ``large'' mosaic there always exist irregular (i.e. T-shaped) nodes. However, at least for $q$ large, irregular nodes may be eliminated by the process, which is the reason for the occurrence of nonphysical solutions mentioned above, see also Section \ref{ss:nonphys}.}
\end{enumerate}
\end{ass}

\begin{figure}[h!]
\begin{center}
\includegraphics[width=\textwidth]{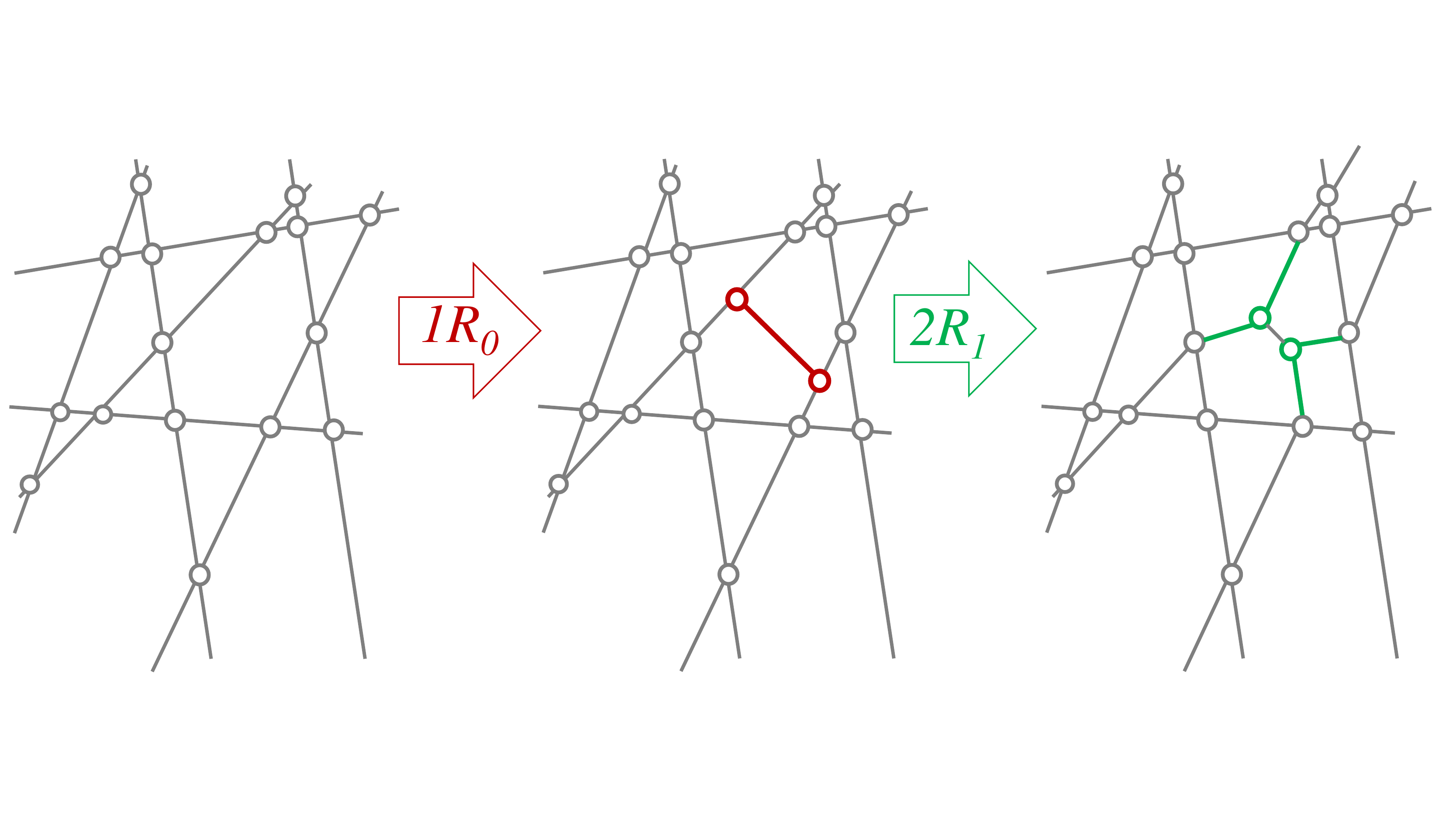}
\caption{Illustration of the micro-events $R_0, R_1$. The associated probabilities in the model presented in \cite{disc} are: $p_0=(1-q) \quad p_1=q$. The label $2R_1$ on the green arrow indicates that the figure corresponds to two crack healing events.}\label{fig:2}
\end{center}
\end{figure}

\begin{rem}
At each micro-event the value of the fundamental variables $N^{\star},F,V$ may remain constant or it may change. Collectively, we will briefly refer to these (possibly zero) changes \emph{increments} and note that they can be classified in two different manners:
\begin{itemize}
\item We may classify them based on the type of micro-event triggering the increment. When doing this, we will use brackets: for example, the increment in the fundamental variable $F$ triggered by the micro-event $R_1$ will be denoted by $\Delta F(1)$.
These increments are \emph{deterministic} and their values are provided,
based on Figure \ref{fig:2}, in equation \eqref{linear_increments} below.
\item We may classify them based on the time order in which they occur, more precisely, we may record the serial number of the micro-event triggering the increment. When doing this, we will use subscripts: for example, the increment in the fundamental variable $F$ triggered by the $k$th micro-event  will be denoted by $\Delta F_k$. As the type of the $k$th micro-event is chosen randomly, these increments are random quantities, in fact, based on Assumption \ref{ass:2}, i.i.d. (independent, identically distributed) random variables. By considering the geometric arrangements in Figure \ref{fig:2}, we may compute their expected values as
\begin{equation}\label{linear_increments}
\begin{array}{rcccccl}
E(\Delta N^{\star}) &  = & \sum_{i=0}^1\Delta N^{\star}(i)p_i & = & 4(1-q)+1q & = & 4-3q \\
 E(\Delta F) &  = &  \sum_{i=0}^1\Delta F(i)p_i & = &  1(1-q)+0q & = & 1-q \\
 E(\Delta V) & =  &  \sum_{i=0}^1\Delta V(i)p_i & = & 2(1-q)+0q & = & 2-2q
\end{array}
\end{equation}
\end{itemize}
It is clear that the above two classifications should not be applied simultaneously: the $k$th micro-event can be either of type $R_0$ or of type $R_1$, but not both.
\end{rem}

\subsection{Emergence of nonphysical solutions.\label{ss:nonphys}}
Obviously, an $R_1$-step can only be performed on an irregular 'T' node.
If we assume the initial mosaic to be regular and we observe that one step of type $R_0$  generates two irregular nodes, the first $k$ pieces of $R_1$ steps must be preceded by at least $(k /2)$ pieces of $R_0$ steps to avoid nonphysical solutions. More generally
we can say that any trajectory with $q > 2/3$ will, sooner or later, exit the domain of geometrically admissible mosaics since
the annihilation rate of 'T' nodes is faster than the rate by which they are generated. We remark that this problem was duly noted in \cite{disc} and the emergence of nonphysical solutions has been described in detail. The problem is rooted in  Assumption \ref{ass:2}(3) which states that
the probabilities $p_0,p_1$ are global constants and the actual state of the mosaics does not affect these numbers. One
of our goals with the expanded model  is to overcome this difficulty and we will do this in Section \ref{sec:math} where Assumption \ref{ass:2} will be replaced by Assumption \ref{ass:3}.

\subsection{Discrete time scales}
In the model presented in \cite{disc} the time elapsing between micro-events was not discussed. Here we describe a setting and related terminology to account for time scales throughout the paper. Simultaneously, to distinguish deterministic from random quantities, we will denote the later by boldface symbols.

Our setting is as follows. Our aim is to describe the evolution of the mosaic between times $t$ and $t+\d t$. The state of the mosaic is deterministic at time $t$, however, by time $t+ \d t$, it's state is random. The time interval $\d t$ is deterministic, on the one hand it is infinitesimally small on some macroscopic time scale, on the other hand, it is asymptotically longer than the average time that elapses between two consecutive micro-events. In section~\ref{sec:math} we give a more precise description of the asymptotic regime in which these time scales can be interpreted.

Let us denote the times of the consecutive micro-events by $(t<)\btau_1<\btau_2<...$ Then we assume that the interarrival times $\d\bT^k=\btau_k- \btau_{k-1}$, $k\ge 1$ (with the convention $\btau_0=t$) are independent and identically, in fact, exponentially distributed. In other words, the $\{\btau_k\}$ form a (homogeneous) Poisson process. The common expected value of the interarrival times will be denoted as  $E(\d\bT)$. Given the time interval $\d t$, introduce the random variable
\[
\bn_{\d}=\max\{k \ge 0 \,|\, \btau_k \le t+\d t \} =\max \{k \ge 0 \,|\, \d\bT^1+\d\bT^2+\dots+\d\bT^k\le \d t \}
\]
that is, the number of micro-events in the time interval $[t,t+\d t]$. By standard properties of the Poisson process, we have
\begin{equation}\label{E_n_delta}
E(\bn_{\d})=\dfrac{\d t}{E(\d\bT)}
\end{equation}

Now we can write a recursion formula
 \begin{equation}\label{discretesystem}
     \begin{array}{rcccl}
     \bF(t+\Delta t) & = & F(t) & + & \bF_{\Delta} \\
     \bV(t+\Delta t) & = & V(t) & + & \bV_{\Delta} \\
     \bN^{\star}(t+\Delta t) & = & N^{\star}(t) & + & \bN^{\star}_{\Delta},
     \end{array}
 \end{equation}
 where
 \begin{equation}\label{deltas1}
       \bF_{\Delta}  =   \sum_{k=1}^{\bn_{\Delta}}\Delta \bF_k, \quad
       \bV_{\Delta}  = \sum_{k=1}^{\bn_{\Delta}}\Delta \bV_k, \quad
       \bN^{\star}_{\Delta}   =  \sum_{k=1}^{\bn_{\Delta}}\Delta \bN^{\star}_k,
\end{equation}
with $\Delta \bF_k$,  $\Delta \bV_k$, $\Delta \bN^{\star}_k$ denoting the increments of $F$, $V$ and $N^{\star}$, respectively, at the $k$th micro-event.

In what follows, we will refer to $F$, $V$ and $N^{\star}$ as the fundamental variables of the model. In section~\ref{sec:math} we will introduce a formalism that allows an extended class of fundamental variables.

\subsection{Problem statement}

In this paper we propose to improve, expand and generalize the model presented in \cite{disc} in the following manner:
\begin{enumerate}
    \item We will replace the a priory probabilities $p_i$ by model-based probabilities, derived from an expanded set of assumptions, which also includes the physical hypothesis about the origin of the micro-events.
    To this end, we will use \emph{exponential clocks}, the number of which will be a function of the fundamental variables.
    \item We will  admit an arbitrary number of fundamental variables which we will denote in a homogeneous manner by $X_i$, $(i=0,1, \dots I)$ and we will use the dimensionless variables $x_i=X_i/X_0$, $(i=1,2, \dots I)$, with $X_0=N^{\star}$.  The $x_i$ correspond to the reciprocals of
    the variables in (\ref{variables}). We will show the mathematical advantages of this representation.
    \item We will derive rigorously the simultaneous $d \to \infty, \Delta t \to 0$ limit to obtain the governing differential equation instead of the recursion formula (\ref{discretesystem})-(\ref{deltas1}).
\end{enumerate}

	\section{Derivation of the general governing differential equation}\label{sec:math}

\subsection{A more general setting for the discrete process}
To admit the inclusion of additional geometric (combinatorial) variables, we re-formulate (\ref{discretesystem}) as
\begin{equation}\label{discrete_general}
\bX_j(t+\Delta t)=X_j(t)+\bX_{\Delta, j}, \qquad (j=0,1,2 \dots J)
\end{equation}
(this translates into (\ref{discretesystem}) under
$ J = 2 ;  X_0 \equiv N;  X_1 \equiv F;  X_2 \equiv V.$)
Next we rewrite (\ref{deltas1}) as
\begin{equation}\label{general_deltas}
\bX_{\Delta ,j}=\sum_{k=1}^{\bn_{\Delta}} (\Delta \bX_j)_k\qquad (j=0,1,2 \dots J)
\end{equation}
and instead of the dimensionless variables introduced in (\ref{variables}) we use their reciprocals (recall our convention $X_0=N^*$):
\begin{equation}\label{general_variables}
x_j=\frac{X_j}{X_0}; \qquad (j=1,2 \dots J).
\end{equation}
Furthermore, let us introduce the notations
\[
\uX=(X_0,X_1,\dots,X_J)\quad \text{and} \quad \ux=(x_1,\dots,x_J)
\]
for the $J+1$ and $J$ component vectors describing the state of the system in the macroscopic and rescaled variables, respectively.

\subsection{Clocks and the clock function}
In the general model, Assumption \ref{ass:2} is generalized and will be replaced by
\begin{ass}\label{ass:3}
The evolution of the variables $X_j$, $(i=0,1,2, \dots J)$ is driven by  micro-events $R_i,  (i=0,1,\dots I)$ of $I+1$ different types.
These events are triggered by $(I+1)$ types of \emph{clocks} which we associate with the fundamental variables on the mosaic
and we express the number $C_i$ of the $i$th clock type as a \emph{linear} function
\begin{equation}\label{clocknumber}
C_i=C_i(\uX)=\sum_{j=0}^{J} C_{i,j}X_j, \quad (i=0,1,2, \dots I).
\end{equation}
The coefficient matrix $C_{i,j}$ is one of the main (deterministic) inputs
of the model.
 The total number of clocks is $\sum_{i=0}^I C_i$, which may be compared (although not equal to) $X_0$. We assume that with increasing diameter $d$, the quantities $X_j$ $(j=0,\dots,J)$, and thus $C_i$ $(i=0,\dots,I)$, grow as $\asymp d^2$.

We discuss clock signals on three different levels. Throughout, we will use the standard property that the minimum of two exponential variables is exponentially distributed with the parameters of the distributions added.\footnote{This is analogous to the property of the
Poisson process that the superposition of two independent Poisson processes is a Poisson process with the intensities added. Later we will see that in the appropriate limit it is indeed relevant to consider Poisson processes.}
\begin{itemize}
    \item Any \emph{individual clock} of type $R_i$ is giving
signals according to an exponential random variable with parameter $\lambda_i$: the individual clock signals are coming at  frequencies  $\lambda_i$.

\item Any \emph{clock type}  (i.e. the set of all clocks of type $R_i$) is giving signals according to a cumulative process at frequency
$ f_i=C_i\lambda _i $. In particular, the interarrival times between consecutive impacts of type $R_i$ are distributed as an exponential random variable $\Delta \bT_i$, where $\frac{1}{E(\Delta \bT_i)}=f_i=C_i \lambda_i \asymp d^2$

\item The set of \emph{all clocks} is giving signals according to cumulative process at frequency
$f=\sum_{i=0}^I f_i=\sum_{i=0}^I C_i \lambda_i$ and has the interarrival times $\Delta \bT$ which are exponential with
\begin{equation}\label{propdeltaTT}
f=\frac{1}{E(\Delta \bT)}= \sum_{i=0}^I f_i= \sum_{i=0}^I C_i \lambda_i \asymp d^2.
\end{equation}
\end{itemize}
 \end{ass}

 We remark that Assumption \ref{ass:2} can be regarded as a special case of the scenarios
 admitted by Assumption \ref{ass:3}.  We will illustrate this in Section \ref{sec:return} where
 we derive the governing equations for both models.

To obtain quantities of order $1$, it makes sense to rescale  by $X_0$. We introduce the \emph{clock densities}
\begin{equation}\label{clockdensity}
c _i=\frac{C_i}{X_0}, \quad (i=1,2, \dots I)
\end{equation}
and the individual \emph{clock functions}
\begin{equation}\label{clockfunction0}
\gamma _i=\frac{f_i}{X_0}=\frac{1}{X_0 E(\Delta \bT_i)} =c_i\lambda _i \quad (i=1,2, \dots I),
\end{equation}
Both of these quantities ($c_i, \gamma _i$) are associated with clock types, i.e. they characterize the set of $R_i$-type clocks.

The ratio $\gamma_{i_1}/\gamma _{i_2}$ of any two clock functions agrees with the ratio of the cumulative frequencies associated with the $i_1$-type and $i_2$-type clocks.
We remark that both $C_i, c_i$ and $\gamma _i$ may depend on the variables $x_j$, $(j=1,2, \dots J)$ so henceforth we will write $\gamma_i(\ux)$.
Now we introduce
\begin{defn}\label{clockfunction}
We call
\[
\gamma(\ux)=\sum_{i=0}^I \gamma_i (\ux)=\frac{1}{E(\Delta \bT)X_0}
\]
the \emph{cumulative clock function} of the evolution defined by (\ref{discrete_general}).
\end{defn}

\subsection{Probabilities and expected values associated with micro-events}
As we will discuss below, it is a consequence of the above setup that, in the appropriate simultaneous limit of $\d t\to 0$; $d\to \infty$ the process can be also regraded as follows. The arrivals of all micro-events follow a cumulative Poisson process of intensity $f$ specified in \eqref{propdeltaTT}. This is a sum of $I+1$ independent Poisson processes of frequencies $f_i$, corresponding to the micro-event of type $R_i$, where $i=0,\dots I$. Accordingly, the event type $R_i$ occurs with respective probability $p_i(\ux)  (i=0,1, \dots I)$ which can be computed as
\begin{equation}\label{probabilities}
p_i(\ux)=\frac{f_i}{f}=\frac{\gamma_i(\ux)}{\gamma(\ux)}.
\end{equation}

If the $k$th micro-event is of type $R_i$ then the
increment of the variable $X_j$ is  $(\Delta \bX_j)_k=\Delta X_j(i)$, where the value $\Delta X_j(i)$ is deterministic.
Accordingly, the expected increments per micro-event are
\begin{equation}\label{expt}
    \nu_j=E(\Delta \bX_j)= \sum_{i=0}^I \Delta X_j(i)p_i.
\end{equation}
where we introduce the notation $\nu_j=E(\Delta \bX_j)$ for the expected values. These may be regarded as the average increment of the given variable computed for one micro-event while the variables  $\bX_{\Delta,j}$ (cf.~\eqref{general_deltas}) describe the increment per time-step $\Delta t$, which consists of $\bn_{\d}$ micro-events, where $\bn_{\d}$ is random, too. We also remark that equation (\ref{expt}) is a straightforward generalization of (\ref{linear_increments}).
    Based on these considerations we introduce
 \begin{defn}\label{def:fund}
    If  $(I+1)$ types of micro-events $R_i$, $(i=0,1,\dots I)$ are driving the evolution of $(J+1)$ variables $X_j$, $(j=0,1,2, \dots J)$ then the \emph{fundamental table} of the process is a $(I+1)\times (2J+3)$ matrix, consisting of $(2J+3)$ column vectors the first $J+1$ of which are the coefficients $C_{i,j}$ (as given in equation (\ref{clocknumber})), the next column contains the individual clock intensities $\lambda _i$, complemented by the $(I+1) \times (J+1)$ matrix where the $(i,j)$ element is $\Delta X_j(i)$, $(i=0,1,2, \dots I).$
    \end{defn}
The fundamental table may be regarded as the deterministic input data necessary to write down the governing differential equation. As we will show, every term in the latter can be computed from the former.

\subsection{Simultaneous limits}
In this system we have two essential discrete parameters: the diameter $d$ and the time-step $\Delta t$.
To obtain an ordinary differential equation instead of \eqref{discretesystem}, we seek a simultaneous limit where
$d \to \infty$ and $\Delta t  \to 0$.  Let us parametrize such a limit by the parameter $\rho$. As we will shortly show, only those limits would perform this task where $d$ is growing faster than the decrease of $\Delta t$, more precisely, we have
\begin{equation}\label{lim1}
\lim _{\rho \to \infty}d(\rho)=\infty; \quad \lim _{\rho \to \infty}\Delta t(\rho)=0; \quad \frac{d}{d\rho}(d(\rho)\Delta t(\rho))=0,
\end{equation}
from which we get
\begin{equation}\label{propdeltat}
    \Delta t \asymp d^{-1}.
\end{equation}

\subsection{Main result: the governing differential equations}
Our goal is to prove

\begin{thm}\label{lem:1}
Assume that at time $t$ the state of a mosaic is given by the variables $\ux=(x_1,\dots,x_J)$. Let $\ubx(t+\d t)=(\bx_1(t+\d t),\bx_2(t+\d t),\dots,\bx_J(t+\d t))$ denote the (random) state of the mosaic at time $t+\d t$. Then, in the simultaneous double limit defined in equation (\ref{lim1}), the random variables
\[
\frac{\bx_j(t+\d t) - x_j(t)}{\d t}; \qquad (j=1,2, \dots J)
\]
converge in probability to
\begin{equation}\label{ode1}
    \frac{dx_j}{dt}  =  \gamma(\ux)\cdot \nu_0(\ux) \cdot (\hat x_j(\ux)-x_j), \quad (j=1,2, \dots J),
\end{equation}
where
$$ \hat x_j (\ux)=\frac{\nu_j(\ux)}{\nu_0(\ux)}.$$
\end{thm}

\begin{rem}\label{r:ode_deriv}
Based on Theorem~\ref{lem:1} we propose the ODE model \eqref{ode1} to describe the evolution of the mosaic in the limit of size $d\to\infty$. However,
the statement of Theorem~\ref{lem:1} does not literally imply that the random trajectory $\{\ubx(t), t\ge 0\}$ converges to the deterministic trajectory $\{\ux(t), t\ge 0\}$.
We consider this as an interesting open problem. \\
Note that a standard consequence of the autonomous system \eqref{ode1} is that the variables $x_j$ and $x_{j'}$ satisfy
\begin{equation}\label{ode3}
    \frac{dx_{j}}{dx_{j'}}=\frac{x_j-\hat x_j(\ux)}{x_{j'}- \hat x_{j'}(\ux)}, \quad (j=1,2, \dots J; \ j'=1,2, \dots J).
\end{equation}
\end{rem}
\begin{rem}
Before we start the proof of Theorem~\ref{lem:1} we point out that every term in equation (\ref{ode1})
can be computed from the
fundamental table, given in Definition \ref{def:fund}: the cumulative clock function
$\gamma (\ux)$ is, using Definition \ref{clockfunction} computed from the individual clock functions $\gamma_i(\ux)$ which,
in turn, are computed via (\ref{clockfunction0}) from the clock numbers $C_i(\uX)$ and the intensities (individual clock frequencies) $\lambda_i$.
Both of the latter quantities are entries in the fundamental table. The values of $\nu_j(\ux)=E(\d \bX_j)(\ux)$
are computed via (\ref{expt}) from the last three columns of the fundamental table containing the entries $\Delta X_j(i)$
and the probabilities $p_i(\ux)$ which, in turn, are computed from the clock functions via (\ref{probabilities}).
\end{rem}

\begin{proof}

The core idea behind the proof is a as follows: in the limit of \eqref{lim1}, the total number of relevant micro-events tends to infinity. These are generated independently by exponential clocks. This leads to the analysis of the properties of some associated random variables, which are Poisson distributed.

Note that the number of micro-events $\bn_{\d}$ is actually random. Nonetheless, from equations (\ref{propdeltaTT}),(\ref{propdeltat}) and \eqref{E_n_delta} we may conclude that
\begin{equation}\label{ndelta}
 E(\bn_{\Delta})=X_0 \cdot \gamma \cdot \d t  \asymp d, \quad \text{thus in particular}  \lim _{\rho \to \infty} E(\bn_{\Delta}) = \infty.
\end{equation}
The quantity $E(\bn_{\Delta})$ will play a central role in our computations.

A key ingredient of the proof is the following Lemma:
\begin{lem}\label{lem:2}
For any $j=0,\dots,J$ we have, in the limit of \eqref{lim1} that
\[
\frac{\bX_{\d,j}}{E(\bn_{\d})} \Plim E(\d \bX_j)=\nu_j
\]
where $\Plim$ denotes convergence in probability.
\end{lem}

Let us first complete the proof of Theorem~\ref{lem:1} assuming Lemma~\ref{lem:2}.
For the finite time increment $\Delta t$  of $x_j$ we write:
\begin{align*}
    \bx_j(t+\Delta t)-x_j(t)&=\frac{\bX_j(t+\Delta t)}{\bX_0(t+\Delta t)}-\frac{X_j(t)}{X_0(t)}=\frac{X_0(X_j+\bX_{\Delta, j})-X_j(X_0+\bX_{\Delta, j})}{X_0(X_0+\bX_{\Delta,0})}=\\
    &=\frac{\bX_{\Delta, j}-x_j\bX_{\Delta,0}}{X_0+\bX_{\Delta,0}}.
\end{align*}
Next, after division by $\d t$, and using \eqref{E_n_delta}, we get
\begin{align}
\frac{\bx_j(t+\Delta t)-x_j(t)}{\Delta t} &= \frac{\bX_{\Delta,j}-x_j\bX_{\Delta,0}}{E(\Delta \bT)E(\bn_{\d})(X_0+\bX_{\Delta,0})}= \nonumber\\
&=\frac{\dfrac{\bX_{\Delta, j}}{E(\bn_{\d})}-x_j\dfrac{\bX_{\Delta, 0}}{E(\bn_{\d})}}{X_0 E(\Delta \bT) + \dfrac{\bX_{\Delta, 0}}{E(\bn_{\d})} \d t} \label{dx_j_over_dt}
\end{align}
Now in the limit of $\rho\to\infty$ (that is, in the sense of \eqref{lim1}) the numerator of the expression \eqref{dx_j_over_dt} converges in probability, by repeated application of Lemma~\ref{lem:2}, to
\[
\nu_j - x_j \nu_0.
\]
The first term in the denominator is equal to $(\gamma)^{-1}$ by Definition~\ref{clockfunction}, while we may again apply Lemma~\ref{lem:2} to conclude that the second term in the denominator converges to $0$ in probability. This then implies
\[
\frac{\bx_j(t+\Delta t)-x_j(t)}{\Delta t} \quad \Plim\quad  \frac{\nu_j - x_j \nu_0}{(\gamma)^{-1}} = \gamma(\ux) \cdot \nu_0(\ux)\cdot (\hat{x}_j (\ux)-x_j)
\]
which completes the proof of Theorem~\ref{lem:1}.
It remains to prove Lemma~\ref{lem:2}.

\begin{pfof}{Lemma~\ref{lem:2}}

Throughout this proof, we will keep the index $j=0,\dots ,J$ fixed. Hence, for brevity of notation, we drop this index and write just $x,X,\nu,\bx,\bX, \bX_{\d}$ instead of $x_j,X_j,\nu_j,\bx_j,\bX_j, \bX_{\d,j}$, respectively.

Given some parameter $\lambda>0$, we will use the notation $\mathbf{W} \sim \mathop{POI}(\lambda)$ to express that the random variable $\mathbf{W}$ is Poisson distributed with parameter $\lambda$. Recall, furthermore, the notations from Assumption~\ref{ass:3}. The proof of Lemma~\ref{lem:2} relies on the following Claim:
\begin{claim}\label{claim:Poi}
We have
\[
\bX_{\d}=\bxi+\bta
\]
where
\[
\bxi=\sum_{i=0}^I \d X(i) \bxi_i, \quad \text{with} \quad  \bxi_i \sim \mathop{POI}(f_i \cdot \d t)
\]
while $\bta$ is some random variable such that
\[
\frac{\bta}{E(\bn_{\d})} \ \Plim 0
\]
in the limit specified by \eqref{lim1}.
\end{claim}

The claim is a consequence of the following observations. Recall that we consider the evolution of
\begin{itemize}
\item[-] a portion of the mosaic that has diameter $d$,
\item[-] within a time interval of length $\d t$.
\end{itemize}
The total increment $\bX_{\d}$ of $X$ arises as a sum of increments corresponding to micro-events of type $i$, with $i=0,\dots,I$. In course of a micro-event of type $i$, the increment of $X$ is $\d X(i)$. Hence it remains to analyze the total number of micro-events of type $i$. These arise from micro-events triggered by $C_i$ clocks, with $C_i\asymp d^2$. We may label these clocks by the pair $(i,m)$ with $i=0,\dots, I$ and $m=1,\dots,C_i$. As these individual clocks are giving signals according to an exponential distribution with parameter $\lambda_i$, \textit{in leading order}, the number of micro-events  triggered, within the infinitesimal time interval $\d t$, by the clock $(i,m)$ is Poisson distributed with parameter $\lambda_i \d t$. Now let
\[
\bxi_i=\sum_{m=1}^{C_i} \bxi_{i,m} \text{with} \quad  \bxi_{i,m} \sim \mathop{POI}(\lambda_i \cdot \d t);\ (m=1,\dots, C_i).
\]
Recalling that $f_i=C_i \lambda_i$, we conclude that $\bxi_i \sim \mathop{POI}(f_i \cdot \d t)$ as the sum of independent Poisson distributed random variables is Poisson distributed, with the parameters added.

The above description is accurate in leading order, however, there are discrepancies, the contribution of which is given by the random variable $\bta$. Below we argue that
\begin{equation}\label{bta_expect}
E(|\bta|)=O(d^2 \cdot (\d t)^2),
\end{equation}
which together with \eqref{ndelta} implies
\[
E\left(\frac{|\bta|}{E(\bn_{\d})} \right)=O(d \cdot (\d t)^2)\to 0
\]
in the limit \eqref{lim1}. The convergence in probability then follows by Markov's inequality.

To see \eqref{bta_expect}, the following phenomena have to be taken into account:
\begin{itemize}
\item For $i=0,\dots, I$ and $m=1,\dots, C_i$ fixed, we can only be certain that the variable $\bxi_{i,m}\sim \mathop{POI} (\lambda_i \d t)$ is equal to the number of micro-events triggered by the clock $(i,m)$ if $\bxi_{i,m}=1$. Indeed, after the micro-event has occurred, the geometry of the mosaic changes, and the clock of this label may disappear. Nonetheless (using the standard notation $\bf{1}_A$ for the indicator of the event $A$), we have
    \[
    E(\bxi_{i,m}\cdot \mathbf{1}_{\{\bxi_{i,m}\ge 2\}}) =O((\d t)^2),
    \]
    giving a total contribution to $\bta$ of expectation $O(d^2 (\d t)^2)$, as $C_i\asymp d^2$.
\item It may also happen that  a new clock, which is not present at time $t$, appears within the time interval $\d t$, and triggers a micro-event. To bound this effect, partition the mosaic into some pieces of constant size (say discs and/or polygons). The number of such pieces is $O(d^2)$. As the mosaic is balanced, the total number of clocks per piece is uniformly bounded. Accordingly, (i) the probability that within one piece a new clock is created is $O(\d t)$, (ii) the conditional expectation of the events triggered by the new clock is again $O(\d t)$, giving a contribution which is $O((\d t)^2)$ in expectation. Summing on the pieces we arrive at a total contribution to $\bta$ of expectation $O(d^2 \cdot (\d t)^2)$, as claimed.
\end{itemize}
 By Claim~\ref{claim:Poi}, to complete the proof of Lemma~\ref{lem:2} we only need to show that
 \begin{equation}\label{E_xi}
 \frac{\bxi}{E(\bn_{\d})} \ \Plim \ \nu=\sum_{i=0}^I \d X(i) p_i = \sum_{i=0}^I \d X(i) \frac{\gamma_i}{\gamma}
 \end{equation}
 in the limit of \eqref{lim1}. Note that $\bxi_i\sim \mathop{POI} (X_0 \gamma_i \d t)$. Also, from \eqref{ndelta}, $E(\bn_{\d}) =  X_0 \gamma \d t \to \infty$. To proceed, we need the following simple fact:
 \[
 \text{If } \bW_{\lambda}\sim \mathop{POI}(\lambda), \text{ then } \frac{\bW_{\lambda}}{\lambda} \Plim 1 \text{ as } \lambda\to \infty.
\]
Indeed, by Chebyshev's inequality
\[
\mathbb{P} \left(\left|\frac{W_{\lambda}}{\lambda}-1\right|>\varepsilon \right) = \mathbb{P} \left(\left|W_{\lambda}-\lambda\right|>\varepsilon \lambda \right) \le \frac{\lambda}{\varepsilon^2 \lambda^2} \to 0
\]
for any $\varepsilon>0$, as $\lambda\to\infty$. It follows that
\[
 \frac{\bxi_i}{E(\bn_{\d})} \ \Plim \ \dfrac{\gamma_i}{\gamma}=p_i
\]
which implies \eqref{E_xi}. This completes the proof of Lemma~\ref{lem:2};
\end{pfof} \\
and hence of Theorem~\ref{lem:1}. \end{proof}

The basic concepts  are illustrated in Figure \ref{fig:1} on an example related to the geophysical problem discussed in Sections \ref{sec:discmodel} and \ref{sec:return}.

\begin{rem}
Figure~\ref{fig:1} illustrates a particular realization of the random process in the time interval $[t,t+\d t]$ for the specific case of the model discussed in Section~\ref{sec:return}. In this case we have $J=1$, with additional notations $X_0=N^{\star}$ and $X_1=F$. Also, we have included the quantities
\begin{equation}\label{eq:realize}
\overline{\Delta \bX}_j=\frac{\bX_{\d,j}}{\bn_{\d}}; \qquad \text{specifically} \qquad \overline{\Delta \bF}=\frac{\bF_{\d}}{\bn_{\d}} \text{ and }
\overline{\Delta \bN}^{\star}=\frac{\bN_{\d}^{\star}}{\bn_{\d}};
\end{equation}
the empirical averages of the increments of $\bX_j$ within the time interval $[t,t+\d t]$. The reason for their inclusion is that these random quantities can be observed in the specific realizations of the process. Note that although both the numerator and the denominator are random in  \eqref{eq:realize}, Lemma~\ref{lem:2} implies that in the limit of \eqref{lim1} we have
\[
\overline{\Delta \bX}_j=\frac{\bX_{\d,j}}{\bn_{\d}} \Plim E(\d \bX_j)=\nu_j.
\]
\end{rem}

\begin{figure}[ht!]
\begin{center}
\includegraphics[width=1.1\textwidth]{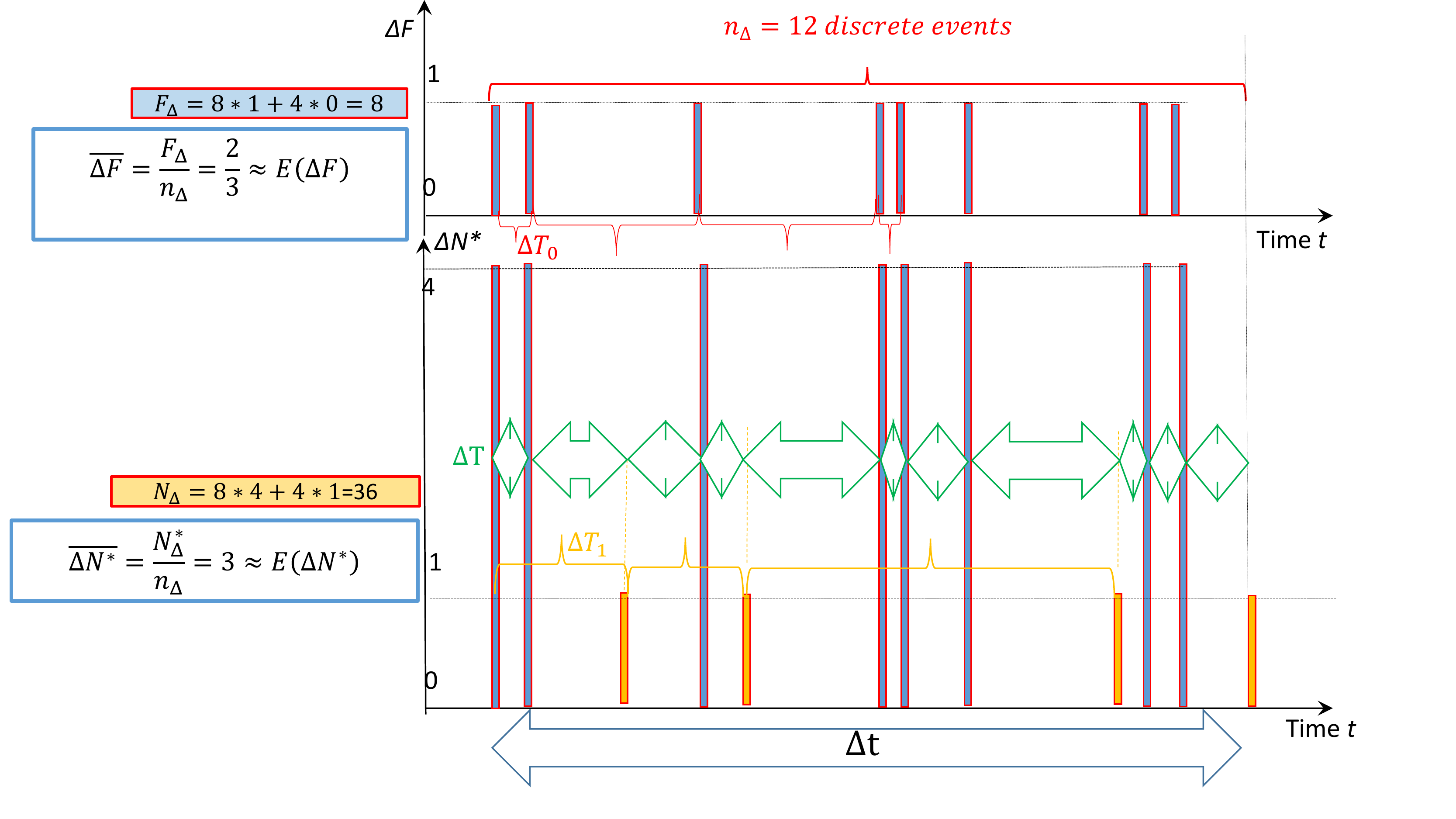}
\caption{Illustration of basic concepts. Instead of the general set of variables $X_i$ we show an example with the variables $N^{\star}\equiv X_1$ and $F \equiv X_2$ used in Sections \ref{sec:discmodel} and \ref{sec:return}. Event of type $R_0$ (with random interarrival time $\Delta T_0$, red color) increases the value of both variables. Event of type $R_1$ (with random interarrival time $\Delta T_1$, yellow color) only increases the value of $N^{\star}$. The cumulative process is characterized by the random interarrival time $\Delta T$ (green color). If both $R_0$ and $R_1$ are governed by an independent Poisson process then the cumulative process is also a Poisson process and the frequencies can be added. }\label{fig:1}
\end{center}
\end{figure}

\begin{defn}
\label{def_odelin}
The case when all of the quantities $\hat{x}_j$, $j=1,\dots,J$ are constant, that is, there exist constants $\hat{x}_j\in\Re^+$ such that
\[
\hat{x}_j(\ubx)\equiv \hat{x}_j,\quad \forall j=1,\dots,J; \text{ and }\forall \ubx\in\Re^J
\]
will be referred to as \textit{the linear case} of the ODE \eqref{ode1}.
\end{defn}

It is important to emphasize that the linear case is a very special, actually, degenerate case of \eqref{ode1}. Below in Section~\ref{ss:ode_lin} we discuss some key properties of the linear case, in particular Lemma~\ref{lem:straight}, which justifies the terminology ``linear''. Later in Section~\ref{sec:return} we revisit crack evolution networks and compare the linear and the nonlinear versions of the model in this particular situation.

\subsection{Discussion of the ODE in the linear case\label{ss:ode_lin}}

Throughout this subsection we consider the special linear case of the ODE \eqref{ode1} in the sense of Definition~\ref{def_odelin}.

\begin{lem}\label{lem:straight} Let us consider the linear case, that is, $\hat x_j =\frac{\nu_j}{\nu_0}=\frac{E(\Delta \bX_j)}{E(\Delta \bX_0)}$ $(j=0,1, \dots, J)$ are all constants. Then, for any choice $j_1,j_j \in \{1,2, \dots J\}; j_1 \not= j_2$ for the  solutions $x_{j_1}(t),x_{j_2}(t)$
of equations $j_1$ and $j_2$ in (\ref{ode1}) we have
\begin{equation}
    \frac{x_{j_1}(t)-\hat x_{j_1}}{x_{j_2}(t)-\hat x_{j_2}}=constant.
\end{equation}
\end{lem}
\begin{proof}
For constant $\hat x_j$, (\ref{ode3}) can be integrated, yielding $d/dt ( ln|x_j-\hat x_j| ) = - c\gamma(x_j) $.
From this we get for  any pair $j_1,j_2$, $j_1\not= j_2$:
$$ d/dt ( ln|(x_{j_1}(t)-\hat x_{j_1})/(x_{j_2}(t)-\hat x_{j_2})| )=0$$
and this gives the statement in the Lemma. We remark that in the $[x_{j_1},x_{j_2}]$ plane these
solutions appear as straight lines.
\end{proof}

\begin{lem}\label{lem:linfixedpoints}
 Consider again the linear case of \eqref{ode1}, that is, the values  $\hat x_j =\frac{E(\Delta \bX_j)}{E(\Delta \bX_0)}>0$ are constant, $j=1,\dots,J$.
 Then  equation (\ref{ode1}) has a global attractor at
 \[
 (x_1,\dots,x_J)=\ux=\hat{\ux}=(\hat x_1,\dots,\hat x_J).
 \]
\end{lem}

\begin{proof}
Based on Definition \ref{clockfunction} we have $\gamma(\ux)>0$. Since the only fixed point is $x_j=\hat x_j$,
this guarantees that it is a global attractor.
\end{proof}

Our next goal is to prove:
\begin{lem}
We call $\rho_{j_1,j_2}(t)=\frac{x_{j_1}(t)}{x_{j_2}(t)}$ the \emph{$(j_1,j_2)$-density} of a mosaic and we call  $\phi_{j_1,j_2} (t)= \frac{\nu_{j_1}}{\nu_{j_2}}=\frac{E(\Delta \bX_{j_1})}{E(\Delta \bX_{j_2})}$ the \emph{$(j_1,j_2)$-fixed point ratio}.  We claim that in the linear case of (\ref{ode1}), when in particular  $\phi_{j_1,j_2}$ is constant, we have that $\rho_{j_1,j_2}(t)$ changes monotonically in the time variable $t$.
\end{lem}

\begin{proof}
To find out about monotonicity we need to compute
\begin{equation}
   sign\left( \frac{d}{dt}\rho_{j_1,j_2}\right) = sign \left( \frac{d}{dt}\frac{x_{j_1}(t)}{x_{j_2}(t)}\right).
\end{equation}
Based on (\ref{ode1}) we find
\begin{equation}\label{sign}
   sign\left( \frac{d}{dt}\rho_{j_1,j_2}\right) = sign(\hat x_{j_1}x_{j_2}-\hat x_{j_2}x_{j_1}),
\end{equation}
which is identical to the sign of the area spanned by the two planar vectors $(\hat{x}_{j_1},\hat{x}_{j_2})$ and $(x_{j_1},x_{j_2})$.
Since, according to Lemma \ref{lem:straight}, the solutions of the $(j_1,j_2)$-system appear as straight half-lines on the
$(x_{j_1},x_{j_2})$-plane the endpoints of which are $(x_{j_1},x_{j_2})=(\hat x_{j_1},\hat x_{j_2})$,
this implies that the sign of the area does not change along the trajectory. This implies
the original claim of the Lemma.

\end{proof}

\section{Governing differential equations for the linear and the nonlinear models}\label{sec:return}

In this section we return to the question of modelling the actual crack evolution phenomena described in Section \ref{sec:intro}. We will show that the general framework introduced in section~\ref{sec:math} provides several possibilities for this modelling task. One of these options, the linear model of subsection~\ref{ss:linear}, can be regarded as the natural continuous time analogue of the discrete time evolution process studied in \cite{disc}. However, we will see that the setup of section~\ref{sec:math} provides also further, more realistic possibilities, specifically the nonlinear model of subsection~\ref{ss:nonlinear}.

\subsection{The linear model}\label{ss:linear}
In \cite{disc} a specific geological model for crack networks is analyzed with the physical meaning of the $J+1=3$ variables given in (\ref{discretesystem}), in particular, we have
\begin{equation}\label{3variables}
N^{\star}\equiv X_0 , \qquad V \equiv X_1, \qquad F \equiv X_2.
\end{equation}
Here we implement this model into the framework of the ODE (\ref{ode1}). The model contains $(I+1)=2$ types of micro-events illustrated in Figure \ref{fig:2}, and we also set $$x \equiv x_1 \equiv \frac{V}{N^{\star}}; \qquad y\equiv x_2 \equiv \frac{F}{N^{\star}}.$$ From Assumption \ref{ass:2} it follows that the number of clocks does not depend on the type of the event but only the frequency does. The fundamental table of this model can be determined based on equation (\ref{linear_increments}). It is given in Table \ref{tab:1:linear} below:
\begin{table}[h!]
    \centering
    \begin{tabular}{| c||c||c|c|c||c||c|c|c|}
         \hline
         $i$ &  Name of step $R_i$  & $C_{i,0}$ & $C_{i,1}$ & $C_{i,2}$  & $ \lambda_i$ & $\Delta N(i)$ & $\Delta V(i)$ & $\Delta F(i)$ \\
         \hline
         \hline
         0 &   Secondary crack & 1 & 0 & 0 & $1-q$ &  4 & 2 & 1 \\
         \hline
         1  & Crack healing &  1 & 0 & 0 & $q$  &  1 & 0 & 0 \\
         \hline
         \hline
    \end{tabular}
    \caption{Fundamental table of the model described in \cite{disc}.}
    \label{tab:1:linear}
\end{table}

First we  compute the clock densities as
\begin{equation}\label{clockdensity1}
c_0=c_1=\frac{N^{\star}}{N^{\star}}=1
\end{equation}
Next we compute the individual clock functions based on (\ref{clockfunction0})
\begin{equation}
    \gamma_0=c_0\lambda_0=1-q; \quad \gamma_1=q
\end{equation}
and, based on Definition \ref{clockfunction}, the cumulative clock function:
\begin{equation}
    \gamma(x,y)=\gamma_0+\gamma_1=1.
\end{equation}
The probabilities can be obtained using (\ref{probabilities}):
\begin{equation}
    p_0=1-q; \quad p_1=q
\end{equation}
and for the expected values $E(\Delta F), E(\Delta V), E(\Delta N^{\star})$ we get, via (\ref{expt}),
the values given in (\ref{linear_increments}).
\begin{equation}
\begin{array}{rccclcl}
    E(\Delta F) & = & p_0\Delta F(0)+p_1\Delta F(1) & = & p_0 &  = & 1-q \\
    E(\Delta V) &= & p_0\Delta V(0)+p_1\Delta V(1) & = & 2p_0 & =& 2-2q \\
    E(\Delta N^{\star}) & = & p_0\Delta N^{\star}(0)+p_1\Delta N^{\star}(1) & = & 3p_0+1 & =& 4-3q.
    \end{array}
    \end{equation}
Next we compute $\hat x,\hat y$ as
\begin{equation}\label{xhatyhat1_linear}
    \hat x=\frac{E(\Delta V)}{E(\Delta N^{\star})}=\frac{2-2q}{4-3q}; \quad \hat y=\frac{E(\Delta F)}{E(\Delta N^{\star})}=\frac{1-q}{4-3q}.
\end{equation}
Substituting (\ref{xhatyhat1_linear}) into (\ref{ode1}) yields the system
\begin{equation}\label{lin1}
 \frac{dx}{dt} = (4-3q)(\hat x -x),
\end{equation}
\begin{equation}\label{lin2}
 \frac{dy}{dt} = (4-3q)(\hat y -y).
\end{equation}
as the governing ODE for the crack evolution model based on the one presented in \cite{disc}. As noted in Lemma \ref{lem:straight}, this system has solution curves appearing as straight lines
in the inverse symbolic plane and, as noted in Lemma \ref{lem:linfixedpoints}, a global attractor arises given by
\begin{equation}\label{eq:linfixedpoints}
    x_c=\hat x \quad y_c=\hat y.
\end{equation}

As noted earlier, for some values of $p$, the fixed point $(x_c(p),y_c(p))$ lies outside $Q$, corresponding to the emergence of nonphysical solutions. See Figure~\ref{fig:4}, Right panel.

\begin{figure}[h!]
\begin{center}
\includegraphics[width=\textwidth]{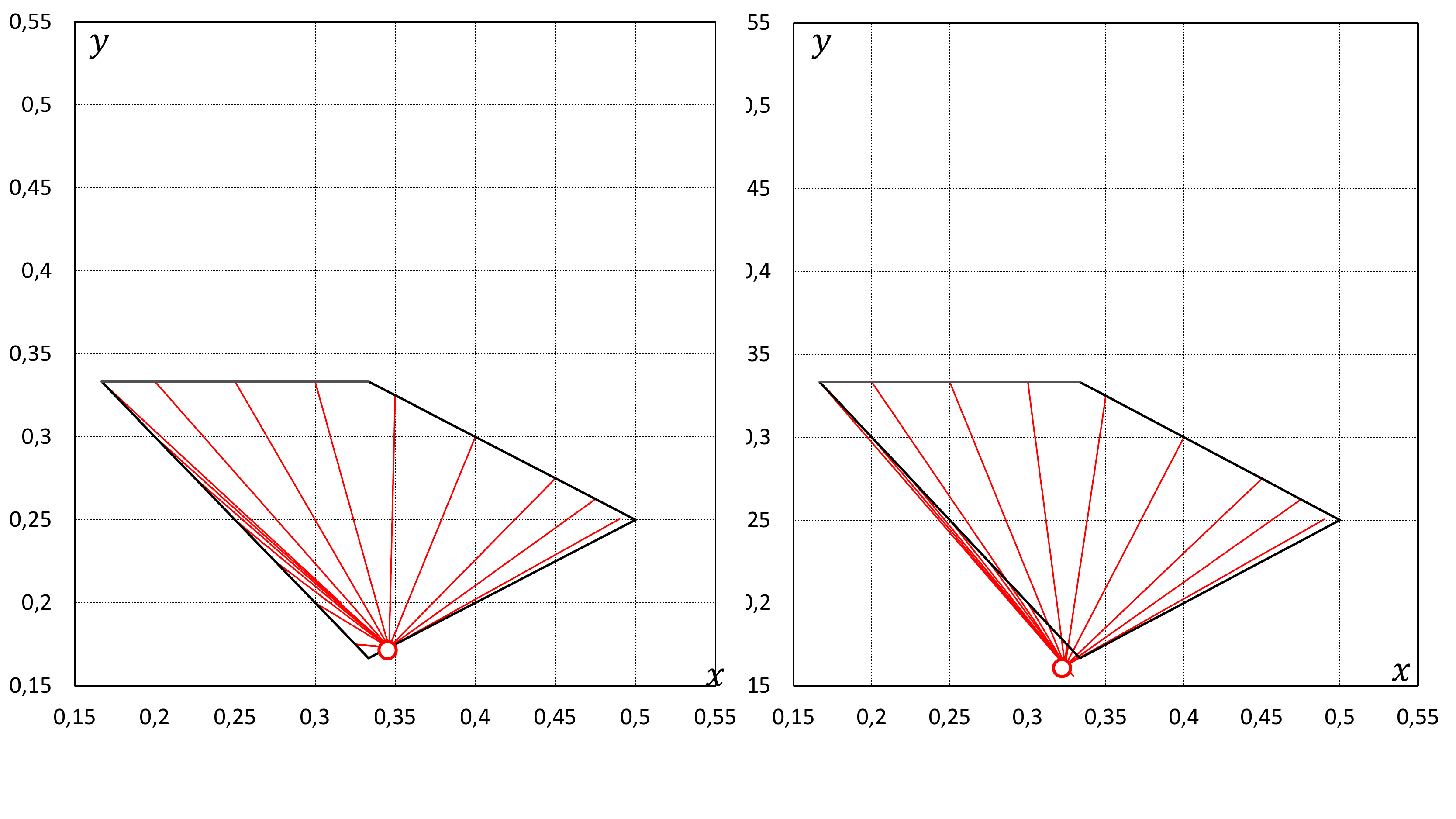}
\caption{Trajectories of the linear model (\ref{lin1})-(\ref{lin2}) on the inverse symbolic plane.
Left panel: $p=0.64$,  fixed point at
$(x_c,y_c)=(0.346, 0.173)$.  Right panel:  $p=0.685$, (nonphysical) fixed point at $(x_c,y_c)=(0.324,0.162)$.}\label{fig:4}
\end{center}
\end{figure}

\subsection{The nonlinear model}\label{ss:nonlinear}
In subsection~\ref{ss:linear} we considered a linear model which is the continuous time analogue of the discrete model of \cite{disc}.
Nonetheless, if instead of Assumption \ref{ass:2} we adopt Assumption \ref{ass:3}, we see that both the number of clocks as well as their frequencies will depend on the type of the event. In the present subsection we build a model for crack evolution that incorporates these possibilities. To get the relevant fundamental quantities,
first we express
the ratio $V_I/V$ of irregular nodes versus all nodes as a function of $x,y$. To this end, we recognize that if we sum up all angles for nodes and we sum up all (internal) angles for cells we must get the same result. We also note that the total angle at a regular node is $2\pi$ while at an irregular node this is $\pi$. So we write
\begin{equation}
V_I\pi +(V-V_I)2\pi=\left(\frac{N^{\star}}{F}-2\right)F\pi,
\end{equation}
from which we get
\begin{equation}\label{vi}
    V_I=2V+2F-N^{\star},
\end{equation}
yielding
\begin{equation}
    \frac{V_I}{V}=2+2\frac{y}{x}-\frac{1}{x}=\frac{2x+2y-1}{x}.
\end{equation}
Using (\ref{vi}), we arrive at the fundamental table given in Table \ref{tab:1}.

\begin{table}[h!]
    \centering
    \begin{tabular}{|c||c||c|c|c||c||c|c|c|}
         \hline
        $i$ & Name of step $R_i$ & $C_{i,0}$ & $C_{i,1}$ & $C_{i,2}$  & $ \lambda_i$ & $\Delta N^{\star}(i)$ & $\Delta V(i)$ & $\Delta F(i)$ \\
         \hline
         \hline
         0 & Secondary crack & 0 & 0 & 1 & $\lambda_0$ &  4 & 2 & 1 \\
         \hline
         1 & Crack healing &  -1 & 2 & 2 & $\lambda_0/\mu$  &  1 & 0 & 0 \\
         \hline
         \hline
    \end{tabular}
    \caption{Fundamental table of the nonlinear model}
    \label{tab:1}
\end{table}
This table has two positive parameters, $\lambda_0$ and $\mu$. $\lambda_0$ is just a generic parameter that scales the speed of the process. $\mu$ is, however, the ratio of the frequencies of the two types of micro-events, and it plays an important role. Specifically, below we investigate how the asymptotic properties of the generated system of ODE depend on the parameter $\mu$.
Now we proceed to compute the clock densities using (\ref{clockdensity})  as
\begin{equation}\label{clockdensity2}
c_0=\frac{F}{N^{\star}}=y; \quad c_1=\frac{V_I}{N^{\star}}=\frac{V_I}{V}\frac{V}{N^{\star}}=2x+2y-1.
\end{equation}
Next we compute the individual clock functions based on (\ref{clockfunction0}) as
\begin{equation}
    \gamma_0=\lambda_0 y; \quad \gamma_1=\lambda _1(2x+2y-1)
\end{equation}
and, based on Definition \ref{clockfunction}, the cumulative clock function:
\begin{equation}\label{gamma}
    \gamma(x,y)=\lambda_1(2x+2y+\mu y-1).
\end{equation}
The probabilities can be obtained using (\ref{probabilities}):
\begin{equation}
    p_0=\frac{\mu y}{2x+2y+\mu y-1}; \quad p_1=\frac{2x+2y-1}{2x+2y+ \mu y-1}.
\end{equation}
and for the expected values $E(\Delta F), E(\Delta V), E(\Delta N^{\star})$ we get, via (\ref{expt}):
\begin{equation}\label{exptvalues}
\begin{array}{rccccclcl}
  \nu _0 & = & E(\Delta N^{\star}) & = & p_0\Delta N^{\star}(0)+p_1\Delta N^{\star}(1) & = & 3p_0+1 & =& \frac{4\mu y+ 2x+2y-1}{2x+2y+\mu y-1} \\
    \nu _1  & = & E(\Delta V) &= & p_0\Delta V(0)+p_1\Delta V(1) & = & 2p_0 & =& \frac{2\mu y}{2x+2y+ \mu y-1} \\
    \nu _2 & = & E(\Delta F) & = & p_0\Delta F(0)+p_1\Delta F(1) & = & p_0 &  = &\frac{\mu y}{2x+2y+\mu y-1}.
    \end{array}
    \end{equation}
Next we compute $\hat x,\hat y$ as
\begin{equation}\label{xhatyhat1}
    \hat x=\frac{\nu _1}{\nu _0}==\frac{2p_0}{3p_0+1}=\frac{2\mu y}{4\mu y+2x+2y-1}; \quad \hat y=\frac{\nu _2}{\nu _0}=\frac{p_0}{3p_0+1}=\frac{\mu y}{4\mu y+2x+2y-1}
\end{equation}
Substituting (\ref{gamma}), (\ref{exptvalues}) and (\ref{xhatyhat1}) into (\ref{ode1}) yields the governing ODE system for the nonlinear crack evolution model:

\begin{equation}\label{nonlin1a}
    \frac{dx}{dt}=\lambda _1 \left(2x+2y+ \mu y-1\right)\frac{4\mu y+ 2x+2y-1}{2x+2y+\mu y-1}(\hat x-x)
\end{equation}
\begin{equation}\label{nonlin2a}
    \frac{dy}{dt}=\lambda_1 \left(2x+2y+ \mu y-1\right)\frac{4\mu y+ 2x+2y-1}{2x+2y+\mu y-1}(\hat y-y).
\end{equation}
Using Remark \ref{convdomain} we find that
\begin{equation}
2x+2y+ \mu y-1=0
\end{equation}
has solutions only outside the domain of convex mosaics, so we arrive at
\begin{equation}\label{nonlin1}
    \frac{dx}{dt}=\lambda _1 (2\mu y-4\mu xy- 2x^2-2xy+x)=f(x,y);
\end{equation}
\begin{equation}\label{nonlin2}
    \frac{dy}{dt}=\lambda_1 (\mu y-4\mu y^2- 2xy-2y^2+y)=g(x,y).
\end{equation}

The main difference between the two models is that Assumption \ref{ass:3} is much closer to the actual geological process and as a consequence, it does not admit nonphysical solutions
exiting the domain of convex mosaics. 

In the subsequent sections~\ref{sss:fixed},~\ref{sss:linstab} and~\ref{sss:globstab} we study the asymptotic properties of the system of ODE \eqref{nonlin1}--\eqref{nonlin2} in the domain of convex mosaics $Q$. In particular, 
we prove the following Proposition.
\begin{prop} \label{p:asymp_nonlin}
For each value of $\mu$, the system of ODE \eqref{nonlin1}--\eqref{nonlin2} has a unique, globally attracting 
fixed point $\ux_c(\mu)=(x_c(\mu),y_c(\mu))$ in $Q$.\\
$\ux_c(\mu)$ is located on the boundary, in particular on the edge $\partial Q_1=Q\cap L_1$ of $Q$ 
(see Remark~\ref{convdomain} for notations). Moreover, the map $\mu\to \ux_c(\mu)$ is one-to-one and onto the interior of the edge $\partial Q_1$; as $\mu$ grows from $0$ to $\infty$, $\ux_c(\mu)=(x_c(\mu),y_c(\mu))$ changes monotonically between the two endpoints of the edge; ie.~between $\left(\frac13,\frac16\right)$ and $\left(\frac12,\frac14\right)$.
\end{prop}
Below the proof of Proposition~\ref{p:asymp_nonlin} is broken down as follows: in Lemma~\ref{lem:nonlinfixedpoints} we find the fixed point $\ux_c(\mu)$, in section~\ref{sss:linstab} a linear stability analysis is performed, and finally in Lemma~\ref{lem:global} we show that for each value of $\mu$, the fixed point $\ux_c(\mu)$ is globally attracting in $Q$.

\subsubsection{Fixed points \label{sss:fixed}}

\begin{lem}\label{lem:nonlinfixedpoints}
The fixed points of (\ref{nonlin1})-(\ref{nonlin2}) in the domain of convex mosaics can be written as
\begin{equation}
    x_c=\frac{2\mu +2}{4\mu +6}=\frac{\mu +1}{2\mu +3}; \quad y_c=\frac{\mu +1}{4\mu +6}.
\end{equation}
\end{lem}
\begin{proof}

To find the fixed points $(x_c,y_c)$ of  \eqref{nonlin1}--\eqref{nonlin2}, we have to solve for $f(x,y)=g(x,y)=0$.
This happens if either
\begin{equation}\label{fixed1}
4\mu y+ 2x+2y-1=0,
\end{equation}
or 
\begin{equation}\label{fixed2}
    \begin{array}{rcl}
    x & = & \hat  x \text{ and}\\
    y & = & \hat y.
    \end{array}
\end{equation}
In the first case, that is, if equation (\ref{fixed1}) holds, we have
\begin{equation}\label{fixed_whydisregard}
y=\frac{0.5-x}{1+2\mu}.
\end{equation}
For $\mu >0$ the line \eqref{fixed_whydisregard} is outside the domain of convex mosaics $Q$ (cf.~Figure~\ref{fig:0}(b)
and Remark \ref{convdomain}). Hence fixed points in $Q$ have to satisfy
(\ref{fixed2}) which
yields, based on (\ref{xhatyhat1}):
\begin{equation}\label{p1}
x= \frac{2p_0}{3p_0+1}; \qquad y=\frac{p_0}{3p_0+1}
\end{equation}
from where we can express $y$ as
\begin{equation}\label{p2}
    y= \frac{\mu y}{4\mu y + 2x + 2y -1}
\end{equation}
from which we get
\begin{equation}\label{p3}
    y= \frac{\mu+1-2x}{4\mu +2}.
\end{equation}
Based on (\ref{p1}) we have $x=2y$ and if we substitute this into (\ref{p3}), we get the expression in the Lemma.
\end{proof}

\begin{rem}
The fixed points of the linear system are given in (\ref{eq:linfixedpoints}) and we can see from (\ref{xhatyhat1_linear}) that for  all fixed points we have $y=x/2$.
Lemma \ref{lem:nonlinfixedpoints} shows that the fixed points of the nonlinear system are located on the same line, however, unlike in the linear system, all fixed points appear in the intervals
$x \in [1/3, 1/2]$ $(y \in [ 1/6, 1/4])$ implying that all trajectories remain inside the domain of convex mosaics.
The trajectories of the linear model are depicted on Figure \ref{fig:4} for the parameters $p=0.64$ and $p=0.68$ (the latter having a nonphysical fixed point). The trajectories of the nonlinear model are depicted on Figure \ref{fig:4a} for the parameters $\mu = 0.125$ and $\mu=1$ .
\end{rem}

\begin{rem}
We note that the edge $\partial Q_1$ of $Q$ -- on which the fixed point $\ux_c(\mu)$ is located for every value of $\mu$ --   corresponds to patterns such that the combinatorial degree of each node is equal to $3$. This is in accordance with our expectations, as the secondary crack steps create nodes with combinatorial degree $3$, while the crack healing steps do not change this property of the nodes involved. Actually, it can be expected that the fixed point $\ux_c(\mu)=(x_c(\mu),y_c(\mu))$ corresponds to a stationary state of the random crack pattern evolution process. In this stationary state nodes have combinatorial degree $3$ with probability $1$, yet, their corner degree is truly random and may take values either $3$ (corresponding to regular, ``Y''-shaped nodes of combinatorial degree $3$) or $2$ (corresponding to irregular, ``T''-shaped nodes of combinatorial degree $3$).
As in course of the process ``T''-shaped nodes are created with rate $\lambda_0$ (secondary cracks) while ``T''-shaped nodes turn into ``Y''-shaped nodes with rate $\lambda_0/\mu$, it is expected that in the stationary state the proportion of ``T''-shaped and ``Y''-shaped nodes are $\frac{\mu}{\mu+1}$ and $\frac{1}{\mu+1}$, respectively. Accordingly we expect that the average corner degree of nodes in the stationary state is
\[
2\cdot \frac{\mu}{\mu+1} + 3\cdot \frac{1}{\mu+1}=\frac{2\mu+3}{\mu+1}
\]
which is exactly $(x_c(\mu))^{-1}$. Hence this prediction is in accordance with the asymptotic properties of the ODE 
\eqref{nonlin1}--\eqref{nonlin2}, which provides further evidence for the relevance of our model.
\end{rem}

\subsubsection{Linear stability analyis \label{sss:linstab}}
Next we perform linear stability analysis of the fixed points. The Jacobian of (\ref{nonlin1})-(\ref{nonlin2})
at the fixed point given in Lemma (\ref{lem:nonlinfixedpoints}) can be computed as:
\begin{equation}
J=\lambda _1
  \begin{bmatrix}
  \frac{-4\mu^2 -10\mu -4}{4\mu +6} & \frac{-4}{4\mu +6} \\
  \frac{-2\mu -2}{4\mu +6} & \frac{-4\mu^2 -6\mu -2}{4\mu +6}
 \end{bmatrix}
\end{equation}
For the eigenvalues and eigenvectors we get
\begin{equation}
\alpha _1= \lambda _1 (-\mu)= -\lambda _0,   \hspace{1cm}
\alpha _2= \lambda _1 (-\mu -1)=-\lambda _0 \frac{\mu +1}{\mu}
\end{equation}
and
\begin{equation}\label{eq:eigen}
v _1=\frac{-u}{\mu +1}
\begin{bmatrix}
1 \\
\mu + 1
\end{bmatrix}   \hspace{1cm}
v _2=u
\begin{bmatrix}
2 \\
1
\end{bmatrix}.
\end{equation}
We first observe that both eigenvalues are negative, so the fixed point locally attracting.
We also observe that the eigenvector $v_2$ is collinear with the $y=x/2$ line at the border of the domain
of convex mosaics. The eigenvector $v_1$ is collinear with the $y=1/2-x$ line if $\mu =0$ and it approaches the vertical direction (parallel to the $y$ axis) as $\mu \to \infty$.
Let us also note that $|\alpha_2|>|\alpha_1|$ for any value of $\mu$, so the attraction in the direction of $v_2$ is always stronger than in the direction of $v_1$.

\begin{figure}[h!]
\begin{center}
\includegraphics[width=\textwidth]{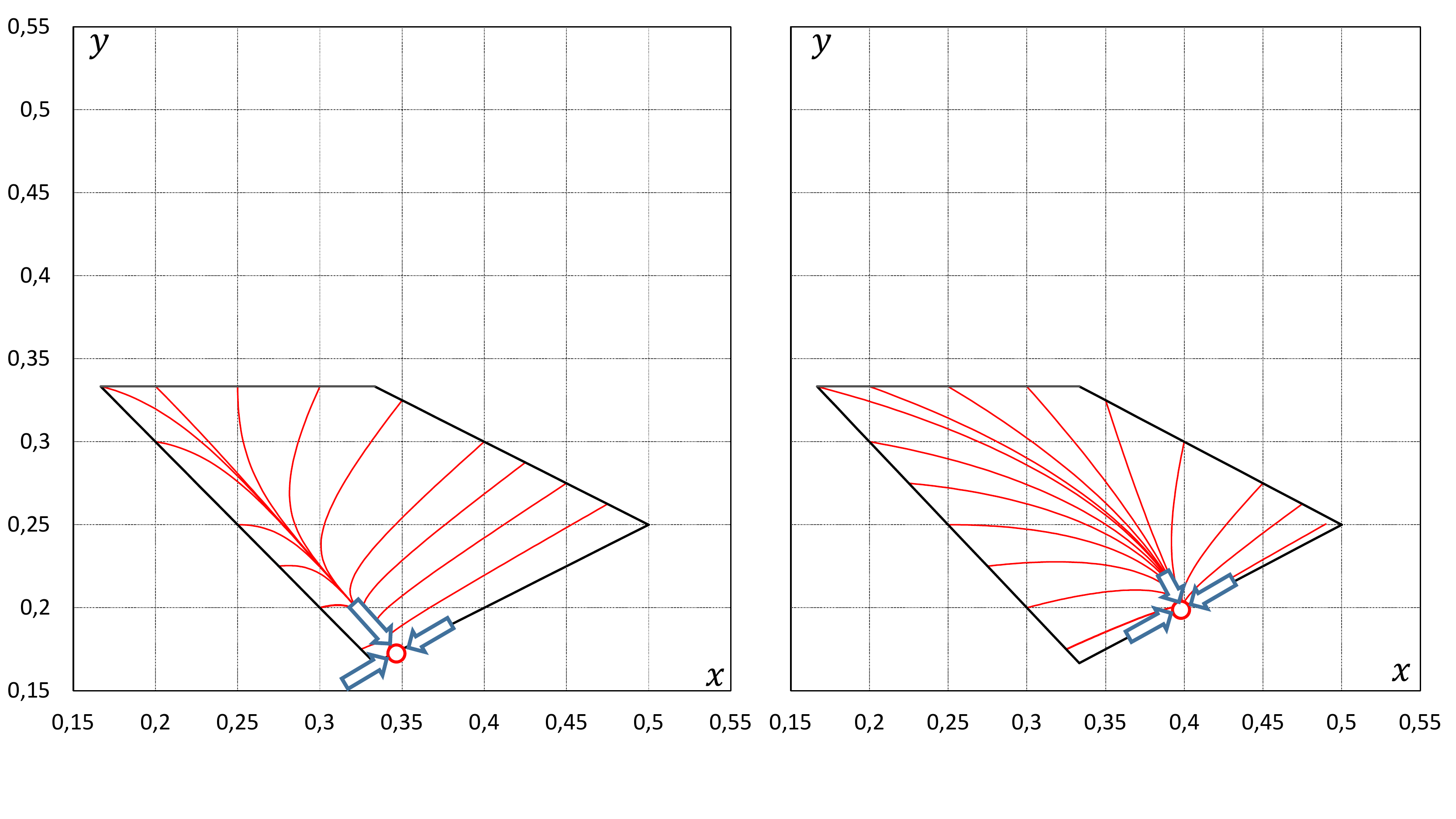}
\caption{Trajectories of the nonlinear model (\ref{nonlin1})-(\ref{nonlin2}) on the inverse symbolic plane. Left panel: $\mu=0.125$, fixed point at
$(x_c,y_c)=(0.346, 0.173)$. Right panel: $\mu=1$, fixed point at
$(x_c,y_c)=(0.4,0.2)$. Eigenvectors (\ref{eq:eigen}) are indicated by blue arrows.}\label{fig:4a}
\end{center}
\end{figure}

\begin{figure}[h!]
\begin{center}
\includegraphics[width=1.5\textwidth]{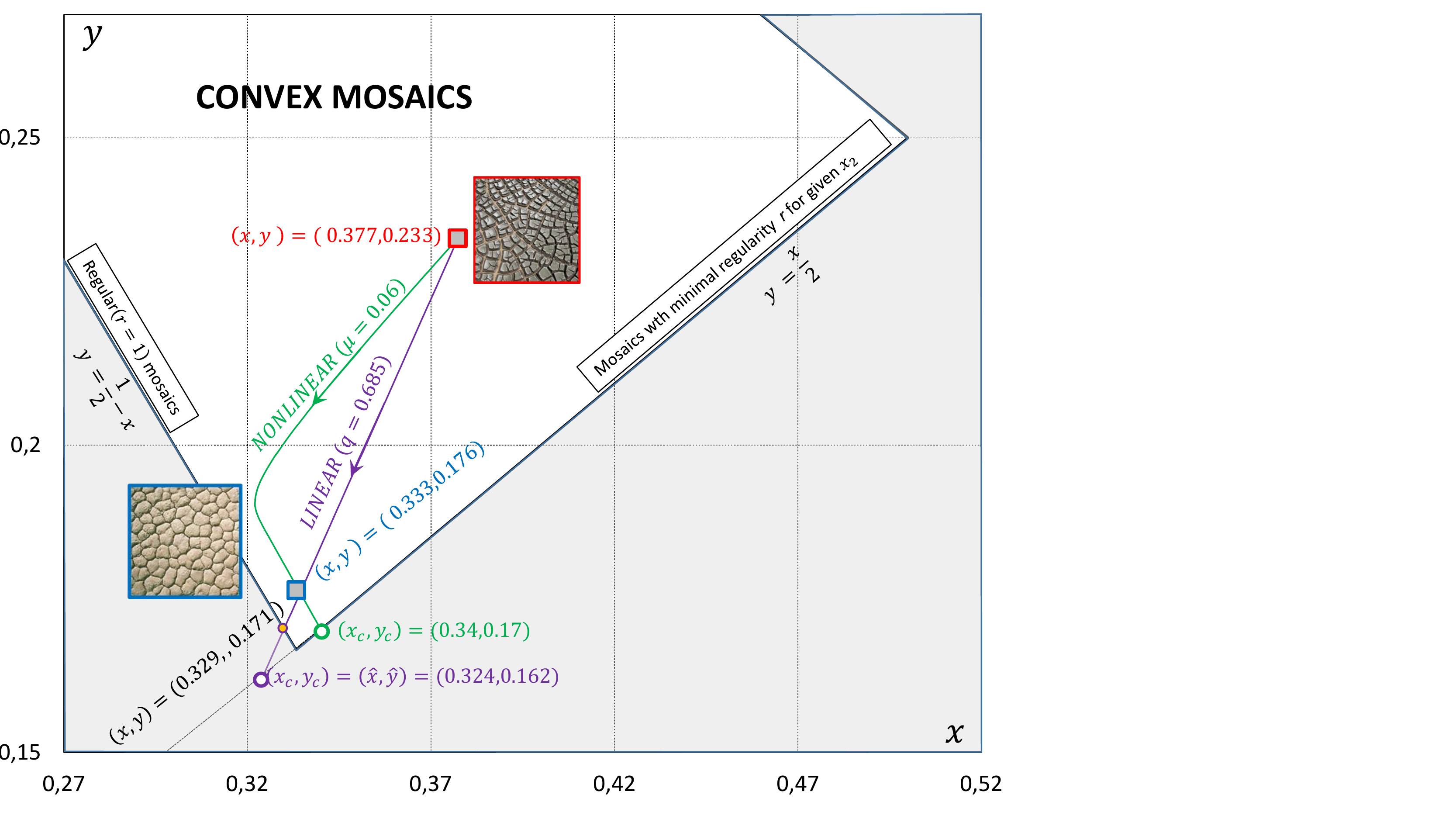}
\caption{Trajectories connecting two geological mudcrack patterns (presented in \cite{Plato}), located in the $[x,y]$ plane at
$(x,y)=0.377,0.233)$ (red frame) and $(x,y)=(0.333,0.176)$ (blue frame). Trajectory for the linear model with parameter $q=0.685$ shown with purple line. Observe that this trajectory transversely crosses the line $y=0.5-x$ marking regular mosaics and also the envelope of convex mosaics at $(x,y)=(0.239,0.171)$ and then continues as a nonphysical trajectory until reaches its limitpoint $ (x_c,y_c)=(0.32,0.16)$ on the $y=x/2$ line. Trajectory from the nonlinear model with parameter $\mu=0.06$ first approaches the $y=0.5-x$ line marking regular mosaics then turns sharply and reaches its limitpoint $(x_c,y_c)=(0.339,0.169)$ on the $y=x/2$ line at the border of the domain
of convex mosaics.}\label{fig:3}
\end{center}
\end{figure}

\subsubsection{Global stability  \label{sss:globstab}}

\begin{lem}\label{lem:global}
 Let $Q$ denote the closed quadrangle of physical configurations described in Remark \ref{convdomain}. Then,
 the fixed point $(x_c,y_c)$ is asymptotically stable in $Q$: for every initial condition
$(x_0,y_0)\in Q$ the system (\ref{nonlin1})-(\ref{nonlin2}) satisfies $\lim_{t\to\infty} (x(t),y(t))=(x_c,y_c)$.
\end{lem}

\begin{proof}
\begin{claim}\label{c:boundary}
$Q$ is forward invariant under (\ref{nonlin1})-(\ref{nonlin2}). That is, for any $(x_0,y_0)\in Q$ and $t>0$, we have $(x(t),y(t))\in Q$.
\end{claim}
\begin{proof}
    Let $\partial Q$ denote the boundary of $Q$, where $\partial Q_i =Q\cap L_i$, with
    the lines $L_i$ defined in Remark \ref{convdomain}, equation (\ref{eq:Q}). We introduce the (not necessarily unit) normal vectors $\mathbf{r}_i$, pointing towards the interior of $Q$ as
    $$\mathbf{r}_1=[-1,2]^T; \mathbf{r}_2=[1,1]^T; \mathbf{r}_3=[0,-1]^T; \mathbf{r}_2=[-1,-2]^T $$
    and we also introduce the vectors composed of the derivatives restricted to the lines $L_i$ as
    $$ \mathbf{d}_i=[f|_{L_i},g|_{L_i}]^T.$$
    Now we compute the scalar products $s_i=\mathbf{r}_i\mathbf{d}_i$ and find
    \begin{equation}\label{bdQ}
        \begin{array}{rrrcl}
             (1) & s_1= & \mathbf{r}_1\mathbf{d}_1 & = & 0,  \\
             (2) & s_2= & \mathbf{r}_2\mathbf{d}_2 & = & \mu(0.5-x),  \\
             (3) & s_3= & \mathbf{r}_3\mathbf{d}_3 & = & \mu/9-1/9+2x/3, \\
             (4) & s_4= & \mathbf{r}_4\mathbf{d}_4 & = & x. \\
        \end{array}
    \end{equation}
    From (\ref{bdQ})(1) we see that $L_1$ is itself invariant. Since the scalar products $s_i$, $i=2,3,4$ along all other boundaries are linear functions of $x$, it is sufficient to show that they are positive at the endpoints of relevant intervals. Indeed we find
    $$s_2(1/6),s_2(1/3),s_3(1/6),s_3(1/3),s_4(1/3),s_4(1/2)>0,$$
    thus we completed the proof of Claim \ref{c:boundary}.
\end{proof}

\begin{claim}\label{c:1}
     We claim that $L_1$ is attractive in the following sense:
Using the notations of (\ref{nonlin1})-(\ref{nonlin2}), for every $\mu>0$ there exists a continuously differentiable function $V:Q\to \Re$  such that
 \begin{enumerate}
\item $V(x,y)>0$ for every $(x,y)\in Q \setminus L_1$, while $V(L_1)=0$,
\item $\frac{\partial V}{\partial x}\cdot f(x,y)+\frac{\partial V}{\partial y}\cdot g(x,y)<0$ for $(x,y)\in Q\setminus L_1$.
 \end{enumerate}
 \end{claim}

 \begin{proof}

 Now we set
\begin{equation}
V(x,y)=2y-x.
\end{equation}
Then, condition (1) is trivially satisfied. For condition (2) we write:
\begin{equation}
\frac{\partial V}{\partial x}\cdot f(x,y)+\frac{\partial V}{\partial y}\cdot g(x,y)=2g(x,y)-f(x,y)=
\lambda_1((x-2y)(4\mu y+2x + 2y -1))<0.
\end{equation}
Since we have $\lambda _1>0$ and $x-2y<0$ we only have to show that
$$4\mu y+2x + 2y -1 >0.$$ The latter is also easy to see, since
$4\mu y+2x + 2y -1 > 2x+2y-1$ and for the latter we have
$$2x+2y-1\ge 0, \mbox{ if } (x,y) \in Q,$$
so we completed the proof of Claim \ref{c:1}.
\end{proof}

From Claim \ref{c:boundary} we know that both $Q$ as well as line $L_1$ are invariant under the flow.
Claim~\ref{c:1} then implies that for any initial condition $(x(0),y(0))\in Q$ the $\omega$-limit set is included in $\partial Q_1$. For any $(x(0),y(0))\in \partial Q_1$ we have $\lim_{t\to\infty} (x(t),y(t))=(x_c,y_c)$. This follows as $(x_c,y_c)$ is a unique (and attracting) fixed point in $\partial Q_1$. Thus we have completed the proof of Lemma \ref{lem:global}.
\end{proof}

\begin{rem}\label{rmk:timeinvert}
An important ingredient of the proof of Lemma~\ref{lem:global} is Claim~\ref{c:boundary} about the forward time invariance of $Q$ under the ODE \eqref{nonlin1}--\eqref{nonlin1}, which is itself a remarkable statement. In particular, it shows that no unphysical configurations emerge under this ODE when starting the system from a physical configuration, that is, from $Q$.\\
Let us comment on the fact that invariance holds only in forward time, and no longer for negative values of $t$. That is, trajectories of the ODE starting from outside $Q$ may enter this region. We would like to point out that this does not have any consequences for the original microscopic crack network evolution process of the mosaic. Indeed, while the system of ODE \eqref{nonlin1}--\eqref{nonlin1} is easy to invert in time (by simply reverting the trajectories), it is not obvious how to specify an inverse
of the original crack network evolution. The former is obtained by averaging the later. Even if some ``natural'' inverse of the microscopic process was specified, there is no reason to expect that such an inverted microscopic time evolution would anyhow relate, after averaging, to the reverted trajectories of the ODE \eqref{nonlin1}--\eqref{nonlin1}.
\end{rem}

\section{Concluding remarks}

\subsection{Comparison between the linear and the nonlinear model}\label{ss:lin_nonlin}

Using (\ref{ode1}), in Subsections \ref{ss:linear} and \ref{ss:nonlinear} we derived the governing equations (\ref{lin1})-(\ref{lin2}) for the linear model and (\ref{nonlin1})-(\ref{nonlin2}) for the nonlinear model, respectively. Now we return to the original geological problem, stated in Section \ref{sec:intro} where we set the goal to find possible evolution paths between geological crack patterns.

For better comparison with \cite{disc}, we use the desiccation mud patterns shown in Figure \ref{fig:00} (a) and (b) as examples; the same patterns were analyzed in \cite{Plato} as well. For the locations of these patterns on the symbolic and inverse symbolic planes, see Figure \ref{fig:0}. To illustrate the possible evolution paths, we will only use the inverse symbolic plane $[x,y]$, appearing in Figure \ref{fig:0} (b).

Both the linear and the nonlinear model discussed in Subsections \ref{ss:linear} and \ref{ss:nonlinear}
have a single scalar parameter: $q$ (in the linear model) and $\mu$ (in the nonlinear model). Fitting a trajectory which is initiated at one of the mosaics and hits the other mosaic is a boundary value problem (BVP). We used these
free parameters to solve this problem numerically. We made the following (numerical) conclusions:
\begin{itemize}
    \item[a)] In both models, the BVP could be solved only in \emph{one direction}, i.e. in both cases we could only solve the BVP with the mud crack mosaic in Figure \ref{fig:0}(b), located at $(x,y)=(0.377,0.233)$ as initial point. We did not find a solution to the BVP in the other direction.
    \item[(b)] In the linear model, the solution corresponds to the parameter $q=0.685$ and it has a (nonphysical) limit point at $(x,y)=(0,32,0.16)$.
    \item[(c)] In the nonlinear model, the solution corresponds to the parameter $\mu=0.06$ and it has a (physical) limit point at $(x,y)=(0.339,0.169)$
\end{itemize}
    These findings are summarized in Figure \ref{fig:3}. From the geological perspective, conclusion (a) appears to be relevant as it suggests that our evolution model defines a hierarchy among natural patterns. It is particularly interesting that, although the linear and nonlinear models are fundamentally different, still, they define the same sequence for these two mosaics.

\subsection{Summary and outlook}\label{sec:sum}
In this paper, after revisiting the discrete time model of \cite{disc} in Section~\ref{sec:discmodel}, we introduced an enhanced model for crack network evolution in Section~\ref{sec:math}. We developed upon the previous approach in three manners:
\begin{enumerate}
    \item By generalizing Assumption \ref{ass:2} to Assumption \ref{ass:3} (including exponential clocks), we created a \emph{physically consistent extension} of the  model in \cite{disc}.
    \item We generalized the setting from 3 variables to an \emph{arbitrary number of variables}, admitting the inclusion of models of other physical processes.
    \item Using this basis, we
    \begin{itemize}
        \item derived the rigorous double limit, defining the \emph{general governing ordinary differential equation (\ref{ode1})} and
        \item defined the \emph{fundamental table} of a given physical process.
        \end{itemize}
    Using these two objects, the model-specific ordinary differential equation can be derived for a broad range of processes.
\end{enumerate}

  In Section \ref{sec:return} we derived the model-specific ordinary differential equations for the physical process presented in \cite{disc}: planar crack patterns evolving under secondary cracking of single fragments and crack healing-rearrangement. We developed the governing equations both for the linear version corresponding to the original scenario of \cite{disc}, and the entirely novel nonlinear version. We compared the
  linear and nonlinear models in Subsection \ref{ss:lin_nonlin}.

From the mathematical perspective, there remain interesting open questions around the derivation of the ODE \eqref{ode1} from the stochastic process generated by the individual clocks. As noted in Remark~\ref{r:ode_deriv}, our main Theorem~\ref{lem:1} states only that, for any fixed $t\ge 0$, in the limit of \eqref{lim1}, the random quantities
$\frac{\bx_j(t+\d t) - x_j(t)}{\d t}$ converge in probability to the deterministic value specified by \eqref{ode1}. It is an open problem in what sense the stochastic process $\{\ubx(t), t\ge 0\}$ converges to the deterministic trajectory $\{\ux(t), t\ge 0\}$.

  These models, in particular the nonlinear version, offer an adequate tool for the study of crack network evolution in \emph{constant environments}. For example, laboratory experiments could be targeted with this model.  One remarkable feature which we proved for the nonlinear model (\ref{nonlin1})-(\ref{nonlin2}) is that it has a single global attractor. It appears to be an interesting question which set of
  micro-steps would result in a model with a similar property? Alternatively, one may ask which
  set of micro-steps might produce qualitatively different global dynamics?

\subsection{Possible generalizations}

We discuss two types of generalizations. First we describe the capabilities of the current mathematical model, beyond describing the geophysical process discussed in \cite{disc}.  Then we outline possible generalizations of the mathematical framework itself.

\subsubsection{Capturing other types of geophysical problems with the current tools}
The mathematical tools developed here could go way beyond the  physical process investigated in \cite{disc}; in fact, they admit several types of generalizations, in particular:
    \begin{itemize}
        \item[(a)] the inclusion of additional steps, for example secondary cracking of multiple fragments,
        \item[(b)] the inclusion of other step-types, for example partial cracking,
        \item[(c)] the analysis of 3D crack patterns,
        \item[(d)] the analysis of the evolution of other geological patterns, for example ridge patterns, illustrated in Figure \ref{fig:00}(c),
        \item[(e)] the analysis of the evolution of other types of discrete patterns, for example, infinite planar graphs or non-convex tilings.
          \end{itemize}

We remark that the entries of the fundamental table were constants in our model. However, implementing the above listed generalizations may result in models where the entries in the fundamental table are functions $h(x,y)$.

\subsubsection{Possible generalizations of the mathematical tools}

In our model, described by the system (\ref{ode1}) of autonomous ordinary differential equations (ODEs),  we regarded the clock frequencies $\lambda_i$ and their ratio $\mu$  as \emph{temporal} constants and the combinatorial averages $x,y$ as \emph{spatial} constants. Needless to say, these are only approximations of real-world scenarios.

The frequencies $\lambda _i$ and their ratio $\mu$ represent the environment. In nature one would not expect the environment to remain constant and in the model this would be reflected in the time variation of the model parameters $\lambda_i, \mu$. In geophysical terms, high values of $\mu$ indicate an environment where rigid cracking is dominant while low $\mu$ values indicate an environment where crack healing/rearrangement is common. In mathematical terms, admitting the time variation of $\lambda_i, \mu$  would yield a non-autonomous system of ODEs which
would, at least formally, be also described by (\ref{ode1}), with added explicit time dependence. While this can be certainly handled numerically, the asymptotic properties of such a non-autonomous ODE would require a separate analysis.

In our model we neglected the spatial variation of the combinatorial averages $x,y$, thus we treated the mosaic as a homogeneous entity. The combinatorial properties of very large mosaics may substantially vary in space
and considering this variation may have major effects on the evolution. If we included this spatial dependence then our model would be generalized to a PDE model. Here, the observation that the spatial averages tend to get more uniform as the mosaic evolves in time would correspond to a diffusion term and this may
indeed be the basic scenario observed in nature.



\end{document}